\documentclass[runningheads,envcountsect]{llncs}

\spnewtheorem{observation}[theorem]{Observation}{\bfseries}{\itshape}
\spnewtheorem{observationapp}[theorem]{$\star$ Observation}{\bfseries}{\itshape}
\spnewtheorem{theoremapp}[theorem]{$\star$ Theorem}{\bfseries}{\itshape}
\spnewtheorem{myclaim}[theorem]{Claim}{\bfseries}{\itshape}
\spnewtheorem{mydefinition}[theorem]{Definition}{\bfseries}{\itshape}
\spnewtheorem{mylemma}[theorem]{Lemma}{\bfseries}{\itshape}
\spnewtheorem{mylemmaapp}[theorem]{$\star$ Lemma}{\bfseries}{\itshape}
\spnewtheorem{mycorollary}[theorem]{Corollary}{\bfseries}{\itshape}

\usepackage{amsthm}
\usepackage{thmtools}

\usepackage{booktabs}
\usepackage{tabularx}

\usepackage{graphicx}
\usepackage{tcolorbox}
\usepackage[ruled,vlined,linesnumbered]{algorithm2e}
\usepackage{amsfonts}
\usepackage{mathtools}
\usepackage{mathrsfs}
\usepackage{xspace}
\usepackage{tcolorbox}
\usepackage{multicol}
\usepackage{amsmath}
\usepackage{comment}
\usepackage{xcolor}
\usepackage{pdfpages}
\usepackage{amssymb}
\usepackage{array}
\usepackage{comment}
\usepackage{calc}
\usepackage{caption}
\usepackage{bbding}
\usepackage{titlefoot}
\usepackage{enumerate}
\usepackage[colorlinks]{hyperref}

\usepackage{etoolbox}
\usepackage[nameinlink]{cleveref}

\usepackage[title]{appendix}
\newcolumntype{L}[1]{>{\raggedright\let\newline\\\arraybackslash\hspace{0pt}}m{#1}}
\newcolumntype{C}[1]{>{\centering\let\newline\\\arraybackslash\hspace{0pt}}m{#1}}
\newcolumntype{R}[1]{>{\raggedleft\let\newline\\\arraybackslash\hspace{0pt}}m{#1}}

\numberwithin{equation}{section}

\DeclareRobustCommand{\rchi}{{\mathpalette\irchi\relax}}
\newcommand{\irchi}[2]{\raisebox{\depth}{$#1\chi$}} % inner command, used by \rchi
\newcommand{\C}{\mathcal{C}}

\newcommand{\Q}{\mathcal{Q}}

\newcommand{\M}{\mathcal{M}}
\newcommand{\F}{\mathcal{F}}

\newcommand{\R}{\mathbb{R}}

\newcommand{\N}{\mathbb{N}}
\newcommand{\Z}{\mathbb{Z}}

\newcommand{\powerset}[1]{\mathscr{P}(#1)}
\renewcommand{\P}{\mathcal{P}}
\renewcommand{\S}{\mathcal{S}}
\newcommand{\order}[1]{\mathcal{O}(#1)}
\newcommand{\ordernopoly}[1]{\tilde{\mathcal{O}}(#1)}
\newcommand{\ordernoinput}[1]{\mathcal{O}^*(#1)}

\newcommand{\yes}{\textup{\textsc{yes}}\xspace}
\newcommand{\no}{\textup{\textsc{no}}\xspace}

\newcommand{\im}[1]{\textup{Im}(#1)\xspace}
\newcommand{\openneighbour}[2]{N_{#1}(#2)}
\newcommand{\closedneighbour}[2]{N_{#1}[#2]}

\newcommand{\pdci}{{\rchi_\delta^S}}
\newcommand{\fpt}{{\sf FPT}\xspace}

\newcommand{\nph}{\textsf{NP}-hard\xspace}

\newcommand{\wih}{\textsf{W}-hard\xspace}
\newcommand{\DS}{\textup{\textsc{Dominating Set}}\xspace}
\newcommand{\HS}{\textup{\textsc{Hitting Set}}\xspace}
\newcommand{\Col}{\textup{\textsc{Graph Coloring}}\xspace}
\newcommand{\ListCol}{\textup{\textsc{List Coloring}}\xspace}
\newcommand{\PreCol}{\textup{\textsc{Pre-Coloring Extension}}\xspace}
\newcommand{\DC}{\textup{\textsc{DomCol}}\xspace}

\newcommand{\DomCol}{\textup{\textsc{Dominator Coloring}}\xspace}
\newcommand{\npsubconp}{\textup{\textsf{NP}}\subseteq\textup{\textsf{coNP}/poly}\xspace}

\newcommand{\CDCol}{\textsc{Class Domination Coloring}\xspace}
\newcommand{\classdcol}{CD coloring\xspace}
\newcommand{\pclassdcol}{partial CD coloring\xspace}
\newcommand{\CDC}{\textup{\textsc{CD Coloring}}\xspace}

\newcommand{\pddci}{{\rchi^S}}

\newcommand{\DCCVD}{\textup{\textsc{DomCol-CVD}}\xspace}

\newcommand{\DCinstance}{(G,\ell)}

\newcommand{\HSinstance}{(U,\F,\kappa)}

\newcommand{\domcol}{dominator coloring\xspace}
\newcommand{\pdc}{partial dominator coloring\xspace}
\newcommand{\pdomcolinstance}{\pdci}

\newcommand{\defproblem}[3]{
	\begin{tcolorbox}[colback=gray!5!white,colframe=gray!75!black]
		\vspace{-1mm}
			\begin{tabular*}{\textwidth}{@{\extracolsep{\fill}}lr} #1   \\ \end{tabular*}
			{\bf{Input:}} #2  \\
			{\bf{Question:}} #3
		\end{tcolorbox}
	}

\Crefname{observation}{Observation}{Observations}
\Crefname{observationapp}{Observation}{Observations}
\Crefname{theoremapp}{Theorem}{Theorems}
\Crefname{mycorollary}{Corollary}{Corollaries}
\Crefname{mylemma}{Lemma}{Lemmas}
\Crefname{mycorollaryapp}{Corollary}{Corollaries}
\Crefname{mylemmaapp}{Lemma}{Lemmas}
\begin{document}
\title{Dominator Coloring and CD Coloring in Almost Cluster Graphs\footnote{A preliminary version of this paper appeared in th 18th Algorithms and Data Structures Symposium (WADS 2023).}}
\author{
	Aritra Banik \inst{1} \and
	Prahlad Narasimhan Kasthurirangan \inst{2} \and
	Venkatesh Raman \inst{3} 
}
\authorrunning{Banik et. al.}
\institute{
	National Institute of Science Education and Research, HBNI, Bhubaneswar, India \email{aritra@niser.ac.in}\\ \and
	Stony Brook University, SUNY, New York, USA
	\email{prahladnarasim.kasthurirangan@stonybrook.edu}
	\and
	The Institute of Mathematical Sciences, HBNI, Chennai, India 		
	\email{vraman@imsc.res.in}
}

\maketitle              % typeset the header of the contribution
% !TEX root = dcoloring.tex

\begin{abstract}
In this paper, we study two popular variants of \Col{} -- \DomCol{} and \CDCol. In both problems, we are given a graph $G$ and a $\ell \in \N$ as input and the goal is to properly color the vertices with at most $\ell$ colors with specific constraints. In \DomCol, we require for each $v \in V(G)$, a color $c$ such that $v$ dominates all vertices colored $c$. In \CDCol, we require for each color $c$, a $v \in V(G)$ which dominates all vertices colored $c$. These problems, defined due to their applications in social and genetic networks, have been studied extensively in the last 15 years. While it is known that both problems are fixed-parameter tractable (\fpt{}) when parameterized by $(t,\ell)$ where $t$ is the treewidth of $G$, we consider strictly structural parameterizations which naturally arise out of the problems' applications. 

We prove that \DomCol is \fpt when parameterized by the size of a graph's \textit{cluster vertex deletion} (CVD) set and that \CDCol is \fpt parameterized by CVD set size plus the number of remaining cliques. En route, we design simpler and faster \fpt algorithms when the problems are parameterized by the size of a graph's \textit{twin cover}, a special CVD set. When the parameter is the size of a graph's \textit{clique modulator}, we design a randomized single-exponential time algorithm for the problems. These algorithms use an inclusion-exclusion based polynomial sieving technique and add to the growing number of applications using this powerful algebraic technique.

\end{abstract}

%\newpage 
% !TEX root = dcoloring.tex
\section{Introduction}
\noindent
Graphs motivated by applications in bio-informatics, social networks, and machine learning regularly define edges between data points based on some notion of similarity. As a consequence, we are often interested in how ``close" a given graph is to a (special type of) \textit{cluster graph} -- a graph where every component is a clique. A popular measure of this ``closeness" is the \textit{cluster editing distance}. A graph $G$ has cluster-editing distance $k$ if it is the smallest number such that there exists a set of $k$ edges whose addition to or deletion from $G$ results in a cluster graph. As an introduction to the extensive literature surrounding this parameter, we refer the reader to  \cite{Bocker2009,Guo2009,Bocker2011,Fomin2014}.

Another popular parameter of this type is the \textit{cluster vertex deletion set size} (CVD set size) \cite{Huffner2010,DouchaKratochvil2012,MajumdarRaman2017,Goyal2018}. A CVD set in a graph is a subset of vertices whose deletion leaves a cluster graph. The CVD set size of a graph is the size of a smallest sized CVD set. Note that a graph with cluster-editing distance $k$ has a CVD set of size $2k$. Thus, CVD set size is a smaller parameter than cluster-editing distance. 

In this paper we use CVD set size to study the (parameterized) complexity of two variants \Col{} -- \DomCol and \CDCol. A \textit{coloring} of a graph $G$ is a function $\rchi \colon V(G) \to C$, where $C$ is a set of \textit{colors}. A \textit{proper coloring} of $G$ is a coloring of $G$ such that $\rchi(u) \neq \rchi(v)$ for all $(u,v) \in E(G)$. The set of all vertices which are colored $c$, for a $c \in C$, is called the \textit{color class} $c$. We sometimes refer to the color $c$ itself as a color class. We let $|\rchi|$ denote $|\im{\rchi}|$, the size of the image of $\rchi$. A vertex $v \in V(G)$ \textit{dominates} $S \subseteq V(G)$ if $S \subseteq \closedneighbour{G}{v}$. A \textit{\domcol} $\rchi$ of $G$ is a proper coloring of the graph such that for all $v \in V(G)$, $v$ dominates a color class $c \in \im{\rchi}$. A \textit{\classdcol} $\rchi$ of $G$ is a proper coloring of the graph such that for all $c \in \im{\rchi}$ there exist a $v \in V(G)$ such that $v$ dominates all vertices in the color class $c$. We are now ready to define our problems of interest.
\defproblem{\textsc{\DomCol} (\DC)}
{A graph $G$; an integer $\ell$}
{Does there exist a \domcol{} $\rchi$ of $G$ with $|\rchi|\leq \ell$?}

\defproblem{\textsc{\CDCol} (\CDC)}
{A graph $G$; an integer $\ell$}
{Does there exist a \classdcol{} $\rchi$ of $G$ with $|\rchi|\leq \ell$?}
We use $\DCinstance$ to denote an instance of both these problems since it will be clear from context which problem we are referring to. While both problems have a rich theoretical history (see \Cref{relatedwork}), \CDC has garnered renewed interest due to its practical applications in social networks and genetic networks -- the problem is equivalent to finding the minimum number of (i)  \textit{stranger groups} with a common \textit{friend} in social network graphs \cite{Chen2014}; and (ii) \textit{gene groups} that do not directly regulate each other but are regulated by a common gene in genetic networks \cite{KlavzarTavakoli2021}.

\subsection{Notations}\label{Notations}
% !TEX root = dcoloring.tex
\subsubsection{Graph Notations}
Let $G$ be a graph. We use $V(G)$ and $E(G)$ to denote the set of vertices and edge of $G$, respectively. Throughout the paper we use $n$ to denote $|V(G)|$.  For a vertex $v$, we use $\openneighbour{G}{v}$ to denote the set of its neighbors and $\closedneighbour{G}{v}$ is defined to be $\openneighbour{G}{v} \cup \{v\}$. For any graph $G$ and a set of vertices $M \subseteq V(G)$, we denote the subgraph of $G$ induced by $M$ by $G[M]$. We let $G - M$, for a $M \subseteq V(G)$, denote $G[V(G) \setminus M]$. A \textit{matching} $\M$ is a subset of edges with no common endpoints. If $(u,v) \in \M$, we let $\M(u) = v$. Most of the symbols and notations used for graph theoretical concepts are standard and taken from \cite{Diestel2012}.

\subsubsection{Parameterized Complexity and Algorithms.} 
The goal of parameterized complexity is to find ways of solving \nph problems more efficiently than brute force: here the aim is to restrict the combinatorial explosion to a parameter that is hopefully much smaller than the input size. Formally, a {\em parameterization} of a problem is assigning a positive integer parameter $k$ to each input instance and we say that a parameterized problem is {\em fixed-parameter tractable} (\fpt) if there is an algorithm that solves the problem in time $f(k)\cdot \vert I \vert ^{O(1)}$, where $|I|$ is the size of the input and $f$ is an arbitrary computable function that depends only on the parameter $k$. We use $\ordernoinput{f(k)}$ to denote the running time  of such an algorithm. Such an algorithm is called an \fpt algorithm and such a running time is called \fpt running time. There is also an accompanying theory of parameterized intractability using which one can identify parameterized problems that are unlikely to admit \fpt algorithms. These are essentially proved by showing that the problem is \wih. We refer the interested readers to books such as \cite{Cygan2015,Fedor2019} for an introduction to the theory of parameterized algorithms.

\subsection{Related Work} \label{relatedwork}
% !TEX root = dcoloring.tex
\DC was introduced by Gera \textit{et al.} in 2006 \cite{Gera2006} while \CDC was introduced by Merouane \textit{et al.} in 2012 \cite{Merouane2015} (the problem was termed \textsc{Dominated Coloring} here). These papers proved that \DC and \CDC are \nph (even for a fixed $\ell \geq 4$). Unlike \Col, both \DC and \CDC can be solved in polynomial time when $\ell=3$ \cite{Chellali2012,Merouane2015}. These problems, which marry two of the most well-studied problems in graph theory -- \Col and \DS, have been studied in several papers in the last 15 years. Results in these papers can be broadly categorized into two.

First, there have been several crucial results which establish lower and upper bounds on the size of an optimal \domcol and \classdcol of graphs belonging to special graph classes. For example, refer papers \cite{Gera2007,Arumugam2011,PandaPandey2015,Henning2015,Bagan2017,Chen2017} for results on \DC and \cite{Chen2014,Merouane2015,Bagan2017,KlavzarTavakoli2021} for results on \CDC. The second (seemingly more sparse) are algorithmic results on these two problems. Even for simple graph classes such as trees, algorithmic results have been hard to obtain -- indeed, after Gera \textit{et al.} showed that \DC can be solved in constant-time for paths in \cite{Gera2006}, it took close to a decade and incremental works \cite{Chellali2012,ChellaliMerouane2012} before a polynomial time algorithm was developed for trees in \cite{MerouaneChellali2015}! It is still unknown if \DC restricted to forests is polynomial time solvable. While \DC and \CDC seem extremely similar on the surface, we note a striking dichotomy in complexity results involving the two problems -- \DC is \nph restricted to \textit{claw-free graphs} while \CDC is polynomial time solvable for the same graph class \cite{Bagan2017}. 
%This inherent hardness of \DC and \CDC have led to work on their parameterized complexity.

The parameterized complexity of \DC and \CDC were first explored by Arumugam \textit{et al.} in 2011 \cite{Arumugam2011} and by Krithika \textit{et al.} in 2021 \cite{Krithika2021} respectively. The authors expressed the problems in \textit{Monodic Second-Order Logic} (MSOL) and used a theorem due to Courcelle and Mosbah \cite{CourcelleMosbah1993} to prove that \DC and \CDC parameterized by $(t,\ell)$, where $t$ is the \textit{treewidth} of the input graph, is \fpt. Their expression of these problems in MSOL immediately also shows (by \cite{CourcelleOlariu2000}) that \DC and \CDC parameterized by $(w,\ell)$, where $w$ is the \textit{clique-width} of the input graph, is \fpt. However, both problems have remained unexplored when viewed through the lens of other \textit{structural parameters} that measure the distance (commonly vertex deletion) from a tractable graph class. Such parameters have become increasing popular in the world of parameterized algorithms since they are usually small in practice. We refer the interested reader to the following survey by Fellows \textit{et al.} for an overview of structural parameterization \cite{Fellows2013} and \cite{Jansen2013,Goyal2018} for its use in studying \Col and \DS. 
Our paper initiates the study of structural parameterizations of \DC and \CDC.

\subsection{Our Results, Techniques, and Organization of the Paper} \label{our results}

As a graph with bounded vertex cover has bounded treewidth, using results from \cite{Arumugam2011,Krithika2021}, it is easy to show that the \DC and \CDC are \fpt parameterized by a graph's \textit{vertex cover}. We give details in \Cref{VC}. Due to its general nature, the algorithm has a large runtime. We design faster algorithms for more natural (and smaller) parameters. Our main results are tabulated in \Cref{resulttable}. Our overarching result is the following: \DC parameterized by CVD set size is \fpt. This is shown through an involved branching algorithm. We also show that \CDC parameterized by $(k,q)$, where $q$ the number of cliques that remain on deleting a CVD (which is of size $k$) is \fpt. We design much faster algorithms for larger parameters (i.e., special CVD sets).

We consider two well-studied parameters of this type -- the size of a \textit{clique modulator} (CLQ) and that of a \textit{twin cover} (TC). In \Cref{clq}, we design randomized algorithms for the two problems which run in $\ordernoinput{c^k}$-time for a small constant $c$ when the parameter is CLQ. We show that our algorithm for \CDC is optimal unless the \textit{Exponential-Time Hypothesis} fails. These algorithms use an inclusion-exclusion based polynomial sieving technique in addition to an exact single-exponential algorithm to solve \DC that we develop in \Cref{exactalg}. We believe that this algebraic method holds great potential for use in other \Col variants.  

\begin{table}[ht]
	\centering
	\begin{tabular}{|c|c|c|c|c|c|}
		\hline
		& Exact& CLQ & TC & CVD Set \\ \hline
		\DC &  $\ordernopoly{4^n}$ $\clubsuit$& $\ordernoinput{16^k}$ $\clubsuit$&    $\ordernoinput{2^{\order{k \log{k}}}}$ $\clubsuit$          &   $\ordernoinput{2^{\order{2^k}}}$ $\clubsuit$           \\ \hline
		\CDC &  $\ordernopoly{2^n}$  \cite{Krithika2021}& $\ordernoinput{2^k}$ $\clubsuit$ &  $\ordernoinput{2^{\order{k \log{k}}}}$ $\clubsuit$& $\ordernoinput{2^{\order{2^kkq\log{q}}}}$ $\clubsuit$\\ \hline
	\end{tabular}
	\caption{A summary of results. Cells marked $\clubsuit$ are proved in this paper.}
	\label{resulttable}
\end{table}

We show that \DC and \CDC admit $\ordernoinput{2^{\order{k \log{k}}}}$-time algorithms when $k$ is the size of a twin cover in \Cref{tc}. For this purpose, we introduce the notion of a \pdc and a \pclassdcol and show that their corresponding extension problems (similar to \PreCol) can be solved quickly. We prove that the extension problem involving \DC can be solved using a relationship between \DC and \ListCol that we establish. On the other hand, we formulate the \CDC extension problem as an \textit{Integer Linear Program} which, in turn, can be solved using well known methods. We then show that an optimal-sized \domcol (resp. \classdcol) can be obtained as an extension of a small number of \pdc{}s (resp. \classdcol{}s). Since optimal CVD sets, TCs, and CLQs can be found quickly \cite{Huffner2010,Ganian2015,Gutin2021}, we implicitly assume that these sets are also given as input. \Cref{lowerbounds} establishes some lower bounds for \DC and \CDC with respect to these parameters.

\section{Parameterized by Vertex Cover Size}\label{VC}
We show that \DC and \CDC parameterized by the size of the input graph's vertex cover (a subset of vertices whose deletion leaves an edgeless graph) is \fpt as a consequence \cite{Arumugam2011} and \cite{Krithika2021} respectively. These papers \cite{Arumugam2011,Krithika2021} showed that  \DC and \CDC parameterized by $(t,\ell)$ are \fpt. Let $M = \{v_1,v_2 \dots v_k\}$ be a vertex cover of size $k$ of the input graph $G$. Note that every isolated vertex of $G$ must get its own color in a valid \domcol (\classdcol) of $G$. Thus, removing all isolated vertices from $G$ and reducing $\ell$ by the number of such vertices gives us equivalent instances of the problems. If $\ell > k$ in a connected graph $G$ then this is a \yes instance of \DC \cite{Gera2006,Henning2015}. Similarly, it is easy to see that if $\ell > 2k$ in a connected graph $G$ then this is a \yes instance of \CDC{} -- color $v_i \in  M$ with $c_i$ and use color $c_{k+i}$ to color a $v \in V(G) \setminus M$ which has $v_i$ as the vertex with smallest index in $\closedneighbour{G}{v} \cap M$. Moreover, since $t \leq k$ \cite{Jansen2013,Cygan2015}, by the results referred to above, we have an algorithm which solves \DC and \CDC in $f(k) \cdot |V(G)|^{\order{1}}$-time for some computable $f$.  
\begin{observation}
	\DC and \CDC, parameterized by the vertex cover number, are \fpt. 
\end{observation}

\section{Exact Algorithm for DomCol} \label{exactalg}
% !TEX root = dcoloring.tex
We present an inclusion-exclusion based algorithm to solve \DC. First, we require some definitions. Let $\mathcal{U}$ be an arbitrary set of elements and $\S \subseteq \powerset{\mathcal{U}}$. We assume that $\S$ can be enumerated in $\ordernopoly{2^{|\mathcal{U}|}}$-time. An $\ell$\textit{-partization} of the system $(\mathcal{U}, \S)$ is a set $\{S_1,  S_2 \dots S_\ell\} \subseteq \S$ such that (i) $\cup^\ell_{i=1}S_i = \mathcal{U}$ and (ii) $S_i \cap S_j = \emptyset$ for $i \neq j$. The following theorem, which forms the bedrock of popular inclusion-exclusion based algorithms, was proved in \cite{Bjorklund2009}. We refer an interested reader to Chapter 4 of \cite{Fomin2010} for a discussion on this algorithmic technique. We use $\ordernopoly{\cdot}$ to hide polynomial terms.
\begin{theorem} \label{inclusion_exclusion}
	One can decide, in $\ordernopoly{2^{|\mathcal{U}|}}$-time and exponential space, if there exists an $\ell$-partization of $(\mathcal{U},\S)$.
\end{theorem}

Given an instance $\DCinstance$ of \DC, we use \Cref{inclusion_exclusion} to decide if this instance is a \yes in $\ordernopoly{4^n}$-time. To the best of our knowledge, this is the first algorithm described for the problem that betters the trivial  $\ordernopoly{n^n}$-time algorithm. 

Consider an instance $\DCinstance$ of \DC. Let $V(G) = \{v_1,v_2 \dots v_n\}$ where the vertices are ordered arbitrarily. Let  $\mathcal{I}_G \subseteq \powerset{V(G)}$ denote the collection of independent sets of $G$ and $V'(G) = \{v'_i\}^n_{i=1}$ be a copy of $V(G)$. Define $\mathcal{U}_G = V(G) \cup V'(G)$ and $\S_G = \{ I \cup \Delta(I) \mid I \in \mathcal{I}_G \text{ and } \Delta(I) \subseteq \{v'_i \in V'(G) \mid \closedneighbour{G}{v_i} \supseteq I\} \}$. That is, $\S_G$ consists of all tuples where the first component is an independent set of $G$ and the second component is a (possibly empty) set of vertices which dominate all the vertices in the first component. Note that $\mathcal{I}_G$ can be constructed in $\ordernopoly{2^n}$-time and therefore $\S_G$ in $\ordernopoly{4^n}$-time.

\begin{restatable}{observation}{ExactAlgEq}\label{equivalence of DC and partization}
	$\DCinstance$ is a \yes instance of \DC if, and only if, there exists an $\ell$-partization of $(\mathcal{U}_G, \S_G)$. 
\end{restatable}
\begin{proof}
	Assume that $\DCinstance$ is a \yes instance of \DC. Let $\rchi$ be a \domcol of $G$ with $|\rchi| \leq \ell$. Order the colors in $\im{\rchi}$ arbitrarily. For a $c_i \in \im{\rchi}$ let $I_i = \rchi^{-1}(c_i)$ and note that it is an independent set of $G$. Let $\Delta'_i = \{v'_j \in V'(G) \mid \closedneighbour{G}{v_j} \supseteq I_i\}$. Let $\Delta_i = \Delta'_i \setminus \cup^{i-1}_{j=1}\Delta'_j$. Since $\Delta'_i$s form a cover of $V(G)$ (every vertex must dominate a color class), $\Delta_i$s form a partition of $V(G)$. Let $S_i = I_i \cup \Delta_i$. Thus, $\{S_1,S_2 \dots S_{|\rchi|}\} \subseteq \S$ is a partization of $\mathcal{U}$. This proves the forward direction of our claim. Now, assume that there is a partization of $\mathcal{U}$ into $S_1, S_2 \dots S_\ell$ where $S_i = I_i \cup \Delta_i$. Color each vertex in $I_i$ with $c_i$. This is a proper coloring since each of the $I_i$s are independent. Since the $\Delta_i$s cover $V'(G)$ every vertex dominates a color class. Thus, $\DCinstance$ is a \yes instance. 
\end{proof}
From \Cref{inclusion_exclusion} and \Cref{equivalence of DC and partization}, we have the following theorem.
\begin{theorem}\label{c^n algo for DC}
	There exists an algorithm running in $\ordernopoly{4^n}$-time and exponential space to solve \DC.
\end{theorem}
We leave, as an open problem, the existence an algorithm which solves \DC in $\ordernopoly{2^n}$-time.

\section{Parameterized by Clique Modulator Size} \label{clq}
% !TEX root = dcoloring.tex

Consider a graph $G$. A subset $M$ of $V(G)$ is a \textit{clique modulator} if $G-M$ is a clique. We use $Q$ to denote this clique and let $k = |M|$. For convenience, we sometimes use $Q$ in place of $V(Q)$. Our goal for this section is to prove the following: there exists randomized algorithms to solve \DC and \CDC in $\ordernoinput{16^k}$ and $\ordernoinput{2^k}$-time respectively. Under the Exponential-Time Hypothesis, neither of these problems can be solved in $2^{o(k)}$-time (\Cref{ETH lowerbound}). 

Our algorithms follow a similar randomized strategy using polynomial sieving and \textit{Schwartz-Zippel Lemma}, as that of \cite{Gutin2021} for \ListCol parameterized by clique modulator. The overall idea of this technique is as follows: we design a \textit{weighted Edmonds matrix} $A$ whose entries are polynomials over $\R$ such that $G$ admits a \domcol (resp. \classdcol) of size $\ell$ if, and only if, a specified monomial containing $2k$-many (resp. $k$-many) variables divides a term $T$ of $\det{A}$. We use \Cref{Polynomial which is not 0 everywhere}, Lemma 2.5 in \cite{Gutin2021}, with $|J|=2k$ (resp. $|J| = k$) followed by \Cref{Random choice} \cite{Zippel1979,Schwartz1980} to determine whether such a $T$ exists in $\det{A}$.
\begin{theorem}\label{Polynomial which is not 0 everywhere}
	 Let $J \subseteq \{x_1,x_2 \dots x_n\}$ and $P(x_1,x_2 \dots x_n)$ be a polynomial over $\R$. Then, there is a polynomial $Q(x_1,x_2 \dots x_n)$ whose evaluation takes $\order{2^{|J|}}$-time such that $Q \not\equiv 0$ if, and only if, $P$ contains a term divisible by $\prod_{x_j \in J}x_j$.
\end{theorem}
\begin{theorem}\label{Random choice}
	Let $P(x_1,x_2 \dots x_n)$ be a non-zero polynomial over a field $\mathbb{F}$ with maximum degree $d$. Then, 
	$
	\Pr{[P(r_1,r_2 \dots r_n) = 0]} \leq \frac{d}{|\mathbb{F}|} 
	$ if $\{r_1,r_2 \dots r_n\}$ is picked randomly from $\mathbb{F}$.
\end{theorem}  

% !TEX root = dcoloring.tex
\subsection{\DC Parameterized by Clique Modulator Size}\label{DC-CLQ}
Assume that we have an instance $\DCinstance$ of \DC where the size of $Q$ is at most $k$. Then, $|V(G)| \leq 2k$ and therefore, by the proof of \Cref{c^n algo for DC}, $\DCinstance$ can be solved in $\order{16^k}$-time.  Now, consider the case where the size of clique is at least $k+1$. Then, in any proper coloring of $G$, there exists a color which is used (exactly once) in $Q$ but not in $M$. An important observation ensues.

\begin{observation}\label{clique vertices dominate a color class}
	If $|Q| > k$, every vertex in $Q$ dominates a color class in any proper coloring of $G$.
\end{observation}

Our main theorem in this subsection is \Cref{DCCLQ is FPT}. We prove this by designing a weighted Edmonds matrix $A$ whose elements are polynomials containing the variables $X = \{x_v \mid v \in M\}$ and $Y = \{y_v \mid v \in M\}$ with the following property: the existence of $x_v$ in a term $T$ of $\det{A}$ will mean that $v$ has been properly colored and $y_v$ would mean that $v$ dominates a color class. We will therefore need to look for a term of $\det{A}$ which is divisible by  $\prod_{v \in M}x_vy_v$.

\begin{theorem}\label{DCCLQ is FPT}
	\DC can be solved in $\ordernoinput{16^k}$-time where $k$ is the size of a clique modulator of the input graph.
\end{theorem}

Let $C=\{c_1,c_2 \dots c_{\ell}\}$ denote a set of $\ell$-many colors. Moreover, let $C' = \{c'_v \mid v \in M\}$ be a set of $k$-many artificial colors. If $|V(G)| > |C \cup C'|$ (that is, $|Q| + k > \ell + k$), this is a \no instance of \DC as vertices in $Q$ must get different colors. Pad $V(G)$ with $(|C \cup C'| - |V(G)|)$-many artificial vertices. Let this set be $V'(G)$. We construct a balanced bipartite graph $B$ with bipartition $(V(G) \cup V'(G), C \cup C')$ by defining its edges as follows. Every vertex in $V(G)$ is connected all vertices in $C$. In addition, $v \in M$ is also connected to the artificial color $c'_v \in C'$ corresponding to it. Finally, every vertex in $V'(G)$ is connected to all vertices in $C \cup C'$. Each edge $(v,c) \in E(B)$ is associated with a $\S_{(v,c)} \subseteq \P(M)^2$ by the following relation:
\begin{itemize}
	\item If $(v,c) \in Q \times C$, $\S_{(v,c)}$ is the collection of sets $S = (S_1,S_2) \in \P(M)^2$ where $S_1 \cup \{v\}$ is an independent set of $G$ \textbf{and} $S_2 = \{u \in M \mid \closedneighbour{G}{u} \supseteq S_1 \cup \{v\}\}$.
	\item If $(v,c) \in M \times C$, $\S_{(v,c)}$ is the collection of sets $S = (S_1,S_2) \in \P(M)^2$ where $S_1$ is an independent set of $G$ containing $v$ \textbf{and} $S_2 = \{u \in M \mid \closedneighbour{G}{u} \supseteq S_1\}$.
	\item If $v$ or $c$ is an artificial vertex or color, $\S_{(v,c)} = \emptyset$.
\end{itemize} 

We now define our matrix $A$ whose entries are polynomials over $\R$ with dimensions $|V(G) \cup V'(G)| \times |C \cup C'|$. Its rows labeled are by $V(G) \cup V'(G)$ and columns by $C \cup C'$. In addition to the sets of variables $X$ and $Y$, let $Z = \{z_{(v,c)} \mid (v,c) \in E(B)\}$ be a set of variables indexed by edges in $E(B)$. For each $(v,c) \in E(B)$, we define:
\[
	P(v,c) = \sum_{S \in \S_{(v,c)}}\left(\prod_{s \in S_1}x_s \cdot \prod_{s \in S_2}y_s\right) \text{ and } A(v,c) = z_{(v,c)}\cdot P(v,c)
\]
	
Here, we assume that the empty product equals 1. All other entries of $A$ are 0. The proof of the following crucial theorem has a similar flavor as Lemma 3.2 of \cite{Gutin2021}.

\begin{restatable}{theorem}{ComputingZeroofDet}\label{Computing Zero of Determinant is Enough}
	$\DCinstance$ is a \yes instance of \DC if, and only if, $\det A$ contains a monomial divisible by $\prod_{x \in X}x \cdot \prod_{y \in Y}y$.
\end{restatable}
\begin{proof}
	As noted in \cite{Gutin2021}, no cancellation happens in $\det A$. Moreover, a perfect matching $\M$ of $B$ contributes $\sigma(\M) \cdot \prod_{(v,c) \in \M}z_{(v,c)}\cdot P(v,c)$ where $\sigma$ is the sign of $\M$ (when it is considered as a permutation).
	
	Assume that $\DCinstance$ is a \yes instance of \DC. Then, there exists a \domcol $\rchi$ of $G$ with $|\rchi| = \ell$. Assume, without loss in generality, that $C$ is the image of $\rchi$. Order the vertices in $G$ arbitrarily, say $\{v_1,v_2 \dots v_{n-k} \}$ for vertices in $Q$ and $\{v_{n-k+1},v_{n-k+2} \dots v_n\}$ for those in $M$. We define a perfect matching $\M$ of $B$ using $\rchi$ as follows:
	\begin{itemize}
		\item If $v_i \in Q$, set $\M(v_i) = \rchi(v_i)$.
		\item If $v_i \in M$ and $\rchi(v_i)$ is unmatched, $\M(v_i) = \rchi(v_i)$. Otherwise, $\M(v_i) = c'_{v_i}$.
		\item If $v_i \notin V(G)$, let $\M(v_i)$ be an arbitrary unmatched color.
	\end{itemize}
	Note that $\M$ is a perfect matching of $B$. For a $c \in C$, let $M_c = M \cap \rchi^{-1}(c)$ and $M^{-1}_c \subseteq M$ be the subset of vertices that dominate the color class $c$ in $\rchi$. Clearly, for a $(v,c) \in \M$ 	with $c \in C$, $(M_c,M^{-1}_c) \in \S_{(v,c)}$. Let $p_X(v,c) = \prod_{v \in M_c}x_v$ and $p_Y(v,c) = \prod_{v \in M^{-1}_c}y_v$. By construction, $p_X(v,c) \cdot p_Y(v,c)$ is a term of $P(v,c)$. Therefore,
	\[
	T = \alpha \cdot \sigma(\M) \cdot \prod_{(v,c) \in \M}z_{(v,c)}\cdot p_X(v,c)\cdot p_Y(v,c)
	\]
	is a monomial of $\det A$ for some $\alpha > 0$. To complete our proof, we show that for every $u \in M$, there exists edges $(v_X,c_X)$ and $(v_Y,c_Y)$ of $\M$ such that $x_u$ divides $p_X(v_X,c_X)$ and $y_u$ divides $p_Y(v_Y,c_Y)$. Fix a $u \in M$. Let $c_X = \rchi(u)$ and $c_Y$ be a color class that $u$ dominates in $\rchi$. If $c_X \notin \rchi(Q)$, set $v_X = u$ and if $c_X \in \rchi(Q)$, set $v_X$ as the vertex in $Q$ that is colored $c_X$. By construction, $(v_X,c_X) \in \M$ and $x_u$ divides $p_X(v_X,c_X)$. Let $v_Y \in V(G)$ be the vertex with $\M(v_Y) = c_Y$. This proves that $\prod_{x \in X}x \cdot \prod_{y \in Y}y$ divides $T$.
	
	Now, assume that $\det A$ contains a monomial $T$ divisible by $\prod_{x \in X}x \cdot \prod_{y \in Y}y$. Then, there exists a perfect matching $\M$ of $B$ which corresponds to this monomial \cite{Gutin2021}. We can assume that $T = \alpha \cdot \prod_{(v,c) \in \M}z_{(v,c)}p(v,c)$ where $\alpha$ is some constant and $p(v,c)$ is a term of $P(v,c)$ for every $(v,c) \in \M$. We define a mapping $\rchi \colon V(G) \to C$ from this matching and prove that it is a \domcol of $G$. For all $v \in V(G)$ such that $\M(v) \in C$, let $\rchi(v) = \M(v)$. Note that all the vertices in the clique and some vertices in the modulator are now colored. Let $v \in M$ be an uncolored vertex in the modulator. Let $v_i$ be the first vertex in our ordering such that $x_v$ divides $p(v_i,\M(v_i))$ -- such a vertex indeed exists by our assumption. Color $v$ with $c$.
	
	First, we prove that $\rchi$ is a proper coloring of $G$. Let $u$ and $v$ be two vertices such that $c=\rchi(u) = \rchi(v)$. If both $u$ and $v$ are vertices in $M$, then $x_u$ and $x_v$ are terms of $p(v',c)$ for some $v' \in V(G)$. By construction, $u$ and $v$ are part of an independent set. Now, assume that $u \in Q$. Then, $v \in M$ as otherwise $\M$ would contain two edges incident on $c$. Hence $x_v$ is in the term $p(u,c)$ which implies that $(u,v) \notin V(G)$ by construction. This proves that $\rchi$ is a proper coloring.
	
	Finally, we show that every vertex of $V(G)$ dominates a color class in $\rchi$.  Since $\rchi$ is a proper coloring of $V(G)$ and $|Q| > k$, by \Cref{clique vertices dominate a color class}, every vertex of $Q$ dominates a color class. Now, consider a vertex $v \in M$. Then, $y_v$ is a term in $p(u,c)$ for some $(u,c) \in \M$. Moreover, by construction, $u \in V(G)$, $c \in C$, and $\emptyset \subsetneq \rchi^{-1}(c) \subseteq \closedneighbour{G}{v}$. Hence, for all $u \in \rchi^{-1}(c)$, $x_u \in p(v',c)$. This proves that $v$ dominates the color class $c$. 
\end{proof}

We apply \Cref{Random choice} to the polynomial $Q$ obtained when \Cref{Polynomial which is not 0 everywhere} is applied to $\det{A}$ and $J = X \cup Y$. Since $|J| = 2k$, we have a randomized algorithm (whose correctness follows from \Cref{Computing Zero of Determinant is Enough}) which runs in $\ordernoinput{16^k}$-time to solve \DC when restricted to instances where the size of the clique is greater than $k$. With a more complicated polynomial sieving method as in \cite{Gutin2021}, we can improve this to an algorithm which runs in $\ordernoinput{4^k}$-time. Combining this result with the discussion preceding \Cref{clique vertices dominate a color class} gives us the proof of \Cref{DCCLQ is FPT}. Note, therefore, that the bottleneck in the running time of \Cref{DCCLQ is FPT} is due to \Cref{c^n algo for DC}.

% !TEX root = dcoloring.tex
\subsection{\CDC Parameterized by Clique Modulator Size}\label{CD-CLQ}
As in \Cref{DC-CLQ}, we setup the bipartite matching framework used by Gutin \textit{et al.}. The bipartite graph $B$ is constructed as in \Cref{DC-CLQ}. However, the collections associated with the edges of $B$ differ. Each edge $(v,c) \in E(B)$ is associated with a $\S_{(v,c)} \subseteq \P(M)$ by the following relation:
\begin{itemize}
	\item If $(v,c) \in Q \times C$, $\S_{(v,c)}$ is the collection of sets $S \subseteq M$ where $S \cup \{v\}$ is an independent set of $G$ \textbf{and} there exists a $u \in V(G)$ such that $\closedneighbour{G}{u} \supseteq S \cup \{v\}$.
	\item If $(v,c) \in M \times C$, $\S_{(v,c)}$ is the collection of sets $S \subseteq M$ where $S$ is an independent set of $G$ which contains $v$ \textbf{and} there exists a $u \in V(G)$ such that $\closedneighbour{G}{u} \supseteq S$.
	\item If $v$ or $c$ is an artificial vertex or color, $\S_{(v,c)} = \{\emptyset\}$.
\end{itemize} 

We now define a matrix $A$ with dimensions $|V(G) \cup V'(G)| \times |C \cup C'|$, with rows labeled by $V(G) \cup V'(G)$ and columns by $C \cup C'$, whose entries are polynomials. Let $X = \{x_v \mid v \in M\}$ be a set of variables indexed by vertices in $M$ and $Z = \{z_{(v,c)} \mid (v,c) \in E(B)\}$ be a set of variables indexed by edges in $E(B)$. For each $(v,c) \in E(B)$, we define a polynomial $P(v,c)$ as follows:
$
P(v,c) = \sum_{S \in \S_{(v,c)}}\left(\prod_{s \in S}x_s\right)
$. Here, we assume that the empty product equals 1. We let $A(v,c)$, for a $(v,c) \in E(B)$, be $z_{(v,c)}\cdot P(v,c)$. All other entries of this matrix are 0.

We take a moment to reflect on the differences between our setup here and that in \Cref{DC-CLQ}. Since we can ensure that a color class is dominated by a vertex by setting up $\S_{(v,c)}$ carefully, we do not require a set of variables which track if every color is dominated by a vertex (as $Y$ in \Cref{DC-CLQ}). Note that this is important in two fronts -- firstly, this speeds up our algorithm to one that runs in $\ordernoinput{2^k}$-time, and secondly ensures that the running time does not depend on $|Q|$. The proof of \Cref{Computing Zeo of Determinant is Enough Domd Call} flows exactly like \Cref{Computing Zero of Determinant is Enough} and Lemma 3.2 of \cite{Gutin2021}.
\begin{theorem}\label{Computing Zeo of Determinant is Enough Domd Call}
	$(G,\ell)$ is a \yes instance of \CDC if, and only if, $\det A$ contains a monomial divisible by $\prod_{x \in X}x$.
\end{theorem}

As a corollary, we can use \Cref{Random choice,Polynomial which is not 0 everywhere} to obtain the main theorem of this section.
\begin{theorem}\label{CDCCLQ is FPT}
	\CDC can be solved in $\ordernoinput{2^k}$-time where $k$ is the size of a clique modulator of the input graph.
\end{theorem}
\section{Parameterized by Twin Cover Size} \label{tc}
% !TEX root = dcoloring.tex

Consider a graph $G$. A subset $M$ of $V(G)$ is a \textit{twin cover} if for all $(u,v) \in E(G)$ either (i) $u \in M$ or $v \in M$ \textbf{or} (ii) $\closedneighbour{G}{u} = \closedneighbour{G}{v}$ (note that (ii) can be rephrased as ``$u$ and $v$ are \textit{true twins}"). This parameter was introduced by Ganian in 2015 \cite{Ganian2015} and has been used extensively in the world of parameterized complexity since. As before, we let $k$ denote the size of a twin cover throughout this section. We use the following observation due to Ganian \cite{Ganian2015}.

\begin{observation} \label{Twin Cover Clique Neighbourhood}
	Let $G$ be a graph. $M \subseteq V(G)$ is a twin cover of $G$ if, and only if, $G - M$ is a cluster graph and, if $\Q = \{Q_i\}^q_{i=1}$ are its connected components, for any $Q_i \in \Q$ and any $u,v \in Q_i$, $\closedneighbour{G}{u} = \closedneighbour{G}{v}$.
\end{observation}

For a $Q \in \Q$ and an arbitrary $v \in Q$, we let $\closedneighbour{G}{Q} = \closedneighbour{G}{v}$ and say that $u \in V(G)$ is dominated by $Q$ if $v$ dominates it. These notions are well defined as a consequence of \Cref{Twin Cover Clique Neighbourhood}. We prove that \DC and \CDC are \fpt parameterized by the size of a twin cover -- specifically, that they both admit algorithms running in $\ordernoinput{2^{\order{k \log{k}}}}$-time.

A \textit{partial coloring} $\rchi^S \colon S \to C$ is a proper coloring of $G[S]$ (note that the superscript of $\rchi$ specifies the domain of the function). A proper coloring of $G$ is an \textit{extension} of $\rchi^S$ if it agrees with $\rchi^S$ on $S$. In the well-studied \PreCol problem (first introduced in \cite{Biro1992}), we are given a partial coloring $\rchi^S$ of a graph $G$ and an $\ell \in \N$, and ask if there exists an extension $\rchi$ of $\rchi^S$ with $|\rchi| \leq \ell$. It is known that \PreCol is \fpt parameterized by the size of a twin cover \cite{Ganian2015}.

We introduce a similar notion of a \textit{\pclassdcol} and a \textit{\pdc} and prove that their corresponding extension problems can be solved quickly -- in \fpt time in case of \CDC and in polynomial time in case of \DC. We will then prove that an optimal \domcol (resp. \classdcol) can be constructed as an extension of an $\order{2^{\order{k \log{k}}}}$-sized set of \pdc{}s (resp. \pclassdcol{}s). Furthermore, we show that such a set can be be constructed in \fpt time. This, indeed, proves the running time of the aforementioned \fpt algorithms for the two problems.
% !TEX root = dcoloring.tex
\subsection{\DC Parameterized by Twin Cover Size} \label{DCTCsect}

We now introduce the notions of a \pdc of a graph and an extension of a \pdc. We then prove that an optimal extension of a \pdc can be found in polynomial time.
\begin{mydefinition}
	A \textit{\pdc} is a tuple $(\rchi^S,\delta)$ where $\rchi^S$ is a partial coloring of $G$ and $\delta \colon V(G) \to \im{\rchi^S}$ such that for all $v \in V(G)$, $v$ dominates the color class $\delta(v)$. We denote $(\rchi^S,\delta)$ by $\pdci$.
\end{mydefinition}
\begin{mydefinition}
	 A \domcol $\rchi$ of $G$ is an \textit{extension} of a partial \domcol $\pdci$, if $\rchi$ agrees with $\pdci$ on $S$, and for all $v\in V(G)$,  $v$ dominates $\delta(v)$ with respect to $\rchi$.
\end{mydefinition}
\begin{theorem}\label{extension in polytime}
	Let $\pdci$ be a \pdc of $G$ where $M \subseteq S$. Then, given an $\ell \in \N$, we can decide in polynomial time if there exists a extension $\rchi$ of $\pdci$ with $|\rchi| \leq \ell$.
\end{theorem}

We use \ListCol, a well known problem related to \Col, to prove \Cref{extension in polytime}. For more information on the (parameterized) complexity of \ListCol, we refer the reader to \cite{Banik2020,Fiala2011,Gutin2021} and the references therein. We now define \ListCol. Let $G$ be a graph and $C = \{c_1, c_2 \dots c_{\ell}\}$ be a set of colors. A \textit{list function} is a function $L \colon V(G) \to \mathscr{P}(C)$. Given $G$, $C$, and a list function $L$, \ListCol asks the following question -- does there exist a proper coloring $\rchi$ of $G$ such that $\rchi(v) \in L(v)$ for all $v \in V(G)$? It is known that \ListCol can be solved in polynomial time when restricted to cliques \cite{Arora2018}. The proof, which uses a simple matching argument, can be readily extended to cluster graphs. 

\begin{restatable}{mycorollary}{ListClusterPoly} \label{lemma:ListClusterPoly}
	\ListCol, restricted to cluster graphs, is polynomial time solvable.
\end{restatable}

Let $\pdci$ be a \pdc of a $G$ with $M\subseteq S$. If $|\pdci| = \kappa > \ell$, then there obviously exists no extension of it which uses at most $\ell$ colors. Assume that $C = \{c_1, c_2 \dots c_ \ell\}$ is the codomain of $\pdci$ and $\im{\pdci} = \{c_1, c_2 \dots c_\kappa\}$.  Note that $H = G - M$ is a cluster graph. We carefully design a list function $L$ on $V(H)$ so that any proper coloring of $H$ which respects $L$ can be used to color $G$ and the coloring so obtained is an extension of $\pdomcolinstance$.

For a $c \in \im{\delta}$, let $V^c_{\delta}$ be the set of vertices in $G$ that are to dominate $c$ -- that is, $V^c_{\delta} = \delta^{-1}(c)$. Define $\closedneighbour{G}{V_\delta^c} = \cap_{v \in V^c_{\delta}} \closedneighbour{G}{v}$. Observe that $\closedneighbour{G}{V_\delta^c}$ is exactly the set of vertices with the following property: if any subset of vertices in $\closedneighbour{G}{V_\delta^c}$ is colored with $c$ then all the vertices in $V_\delta^c$ will continue dominating the color class $c$. For ease of notation, for a $c \notin \im{\delta}$, we let $\closedneighbour{G}{V_\delta^c} = V(G)$. We now construct $L$. For a $v \in V(H) \cap S$, let $L(v) = \{\pdci(v)\}$. For $v \in V(H) \setminus S$, set $L(v)= \{c \mid v\in \closedneighbour{G}{V^c_\delta}\} \cup \{c_{\kappa + 1}, c_{\kappa+2} \dots c_{\ell}\}$. By construction, we have the following observation.

\begin{observation}\label{ListCol is easy => Ext is easy}
		There exists a list coloring of $H$ respecting $L$ if, and only if, there exists a extension $\rchi$ of $\pdci$ with $|\rchi| \leq \ell$.
\end{observation}

Hence, by \Cref{ListCol is easy => Ext is easy} and \Cref{lemma:ListClusterPoly}, we have the proof of \Cref{extension in polytime}.

Assume that there exists $\Gamma$, a collection of \pdc{}s, with the following properties: 
\begin{enumerate}[(i)]
	\item For all $\pdci \in \Gamma$, $M \subseteq S$.
	\item Given a \domcol{} $\rchi'$ of $G$, there exists an extension $\rchi$ of a $\pdci \in \Gamma$ with $|\rchi| \leq |\rchi'|$.
	\item $|\Gamma| = \order{2^{\order{k \log{k}}}}$ and $\Gamma$ can be constructed in $\ordernoinput{2^{\order{k \log{k}}}}$-time.
\end{enumerate}  
Then, we have the main theorem of this section.
\begin{theorem}\label{dctcfpt}
	Assume that a $\Gamma$ satisfying (i), (ii), and (iii) exists. \DC can be solved in $\ordernoinput{2^{\order{k\log{k}}}}$-time, where $k$ is the size of a twin cover of the input graph. 
\end{theorem}
\begin{proof}
	Consider a collection of \pdc{}s $\Gamma$ with the three properties elucidated above. By \Cref{extension in polytime}, for each $\pdci \in \Gamma$ we can decide in polynomial time if there exists an extension $\rchi$ of $\pdci$ with $|\rchi| \leq \ell$. By the third property of $\Gamma$, we can decide in $\ordernoinput{2^{\order{k\log{k}}}}$-time if there exists an extension $\rchi$ of one of the \pdc{}s in $\Gamma$ with $|\rchi| \leq \ell$. By its second property, $(G,\ell)$ is a \yes instance if, and only if, such an extension exists.
\end{proof}

All we are left to do is to prove that such a collection of \pdc{}s exist. We first prove a few results that will turn out to be useful in the forthcoming paragraphs.

\begin{observation}\label{Twin Cover Clique Vertex Domination}
	Let $v \in Q_i$, for some $Q_i \in \Q$, be a vertex of $G$. Then, $v$ cannot dominate a color that has been used in both $M$ and $V(G) \setminus M$ in any \domcol of $G$.
\end{observation}
\begin{proof}
	Let $c$ be a color that has been used to color a $w \in M$ and $u \in V(G) \setminus M$. If $u \in Q_j$, for $i \neq j$, then $v$ cannot dominate $c$. Now, assume that $u \in Q_i$. Since a \domcol is a proper coloring of $G$, $w \notin \closedneighbour{G}{u}$. Since $M$ is a twin cover, $\closedneighbour{G}{u} = \closedneighbour{G}{v}$. Hence, $w \notin \closedneighbour{G}{v}$. Therefore, $v$ cannot dominate $c$. 
\end{proof}
\begin{mycorollary}\label{Twin Cover Unique Colors Required}
	Let $v \in Q$, for some $Q \in \Q$, be a vertex of $G$. If, in a \domcol $\rchi$ of $G$, $v$ does not dominate any color that is used only to color vertices in $M$, then there exists a vertex $u \in Q$ which gets a color $c$ such that $|\rchi^{-1}(c)|=1$.
\end{mycorollary}
\begin{observation}\label{switching twins domcol}
	Let $\rchi$ be an arbitrary \domcol of $G$ and let $u,v \in V(G)$ such that $\closedneighbour{G}{u} = \closedneighbour{G}{v}$. Then, $\rchi'$, the mapping which switches the colors assigned to $u$ and $v$ by $\rchi$ and retains the coloring of the rest of the vertices, is a \domcol of $G$ with $|\rchi| = |\rchi'|$.
\end{observation}

\subsubsection{Constructing a $\Gamma$}

We now show that such a collection of \pdc{}s exists by describing its construction in two steps. Let $\mathscr{P}$ denote the set of partitions of $M$. For each partition $\P = \{P_1, P_2 \dots P_\kappa\} \in \mathscr{P}$ ($\kappa$ is at most $k$), check if each part $P_i \in \P$ is independent. If so, color all vertices in each $P_i$ by $c_i$. If not, reject this partition of $M$. Note that we have now obtained a collection $\Gamma'$ of $\order{2^{k \log{k}}}$-many partial colorings of $G$. Moreover, constructing $\Gamma'$ required $\ordernoinput{2^{k \log{k}}}$-time. Note that the domain of each partial coloring $\rchi^S \in \Gamma'$ is $M$ by construction.

Consider a $\rchi^M \in \Gamma'$. Let $\mathscr{P}'$ be the set of ``partitions" of $\im{\rchi^M}$ into two sets. Partitions is in quotations here since we allow for $\im{\rchi^M}$ to be split into $\{\im{\rchi^M},\emptyset\}$. Consider a $\P' = \{P'_0,P'_1\} \in \mathscr{P}'$. We interpret colors in $P'_0$ as those that are only used in $M$ and those in $P'_1$ as those that can also be used in $G - M$.

\begin{figure}[ht]
	\centering
	\includegraphics[width=0.67\textwidth]{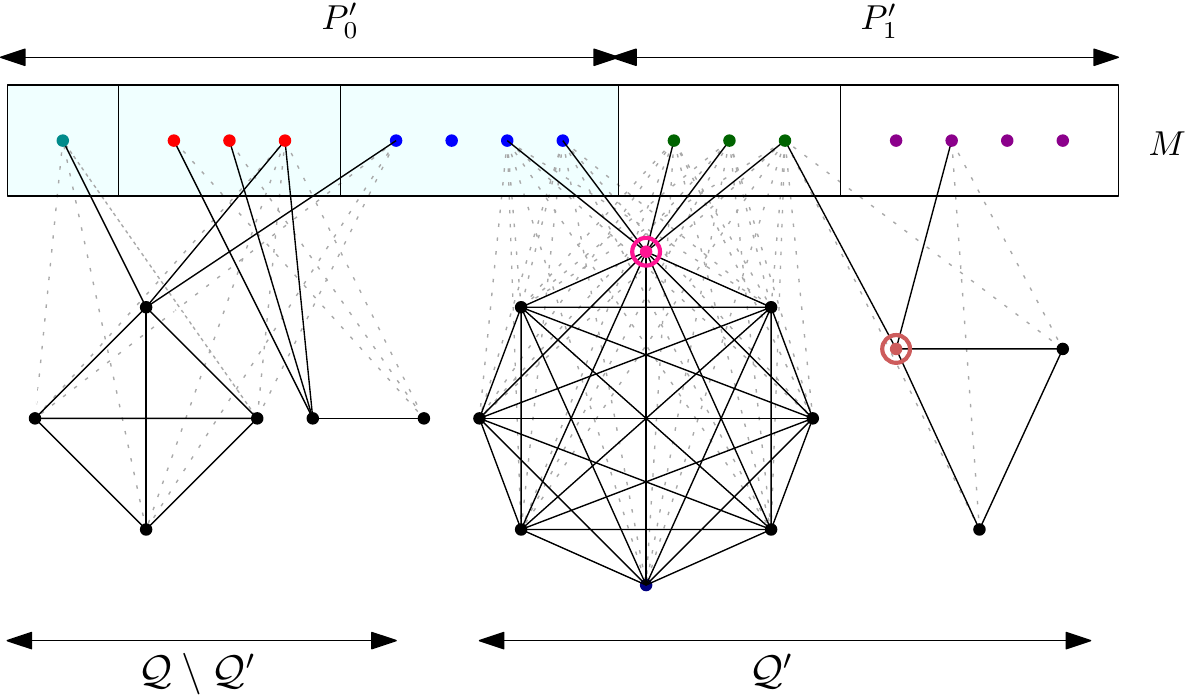}
	\caption{Given a $\rchi^M \in \Gamma'$ and a $\{P'_0,P'_1\} \in \P'$, we define $\Q'$ to be the collection of cliques which do not dominate a color class in $P'_0$. An arbitrary vertex in each of these cliques is colored a unique color.}
	\label{fig:twin1}
\end{figure}  

Let $\Q' \subseteq \Q$ denote the set of cliques in $G-M$ which \textbf{do not} dominate any set in $\{{\rchi^M}^{-1}(c) \mid c \in P'_0\}$. That is, $\Q'$ is exactly the set of cliques that do not dominate a color class that is used exclusively in $M$ (refer \Cref{fig:twin1}). In each $Q \in \Q'$, color an arbitrary vertex $v \in Q$ with a unique color (this is motivated by \Cref{Twin Cover Unique Colors Required} and \Cref{switching twins domcol}). Let the set of these colors be $C_{\Q'}$. 

To obtain a \pdc from this partial coloring, we need to define a domination function $\delta$. At this point, note that every vertex in $G - M$ dominates a color class! A $v \in Q$, when $Q \in \Q'$, dominates the unique color used in that clique. Let $\delta(v)$ be that color. When $Q \notin \Q'$, $v$ dominates a color class in $P'_0$. Let $\delta(v)$ be an arbitrary such color in $P'_0$. Since $|M| = k$, we can assume, without loss in generality, that the color class that a $v \in M$ dominates must be in $\widetilde{C} = \{c_1, c_2 \dots c_\kappa, c_{\kappa + 1} \dots c_{\kappa + k}\}$. Consider the following subset of $M$: 
\[
	M_0 = \{v \in M \mid v \text{ dominates a color class in }P'_0 \cup C_{\Q'}\}
\] 
For a $v \in M_0$, let $\delta(v)$ be an arbitrary color in $P'_0 \cup C_{\Q'}$ that $v$ dominates. By definition, a vertex in $M \setminus M_0$ can only dominate a color class used in $P'_1$ or a color class that has not been used to color any vertex yet. We define the following set of functions:
\[
	\Delta' = \{\delta' \colon M \setminus M_0 \to \widetilde{C} \mid v \text{ dominates all vertices colored }\delta'(v)\}
\]
Note that the above definition allows for a vertex in $M$ to dominate a color class which has not been used to color any vertex in $G$ yet. Fix a $\delta' \in \Delta'$. Let $M_1 = \{v \in M \mid \delta'(v) \in P'_1\}$ and $M_2 = \{v \in M \mid \delta'(v) \notin P'_1\}$. Clearly, $M = M_0 \cup M_1 \cup M_2$. For a $v \in M_1 \cup M_2$, we let $\delta(v) = \delta'(v)$. 

We have one last step to complete before we have a \pdc of $G$ -- we must color at least one vertex (in $G - M$) with each color in $\delta(M_2)$. Consider the bipartite graph $B$ with bipartition $(V(\Q),\delta(M_2))$ where $(v,c) \in E(B)$ if, and only if, $v \in \closedneighbour{G}{V^c_\delta}$ (notation defined prior to \Cref{ListCol is easy => Ext is easy}). If there exists no matching saturating $\delta(M_2)$ in $B$, we do not proceed further with this choice of $\delta'$. If such a matching exists, use each $c \in \delta(M_2)$ to color its matched vertex. If $S$ is the set of vertices of $G$ that have been colored, note that $\pdci$ is indeed a \pdc of $G$. Let $\Gamma$ be the collection of \pdc{}s obtained by varying over $\rchi^M \in \Gamma'$, $\P' \in \mathscr{P}'$, and $\delta' \in \Delta'$.

\subsubsection{$\Gamma$ Satisfies the Three Properties}

Since $|\mathscr{P}'| \leq 2^k$ and $|\Delta'| \leq 2^{k \log{k}}$, $|\Gamma| \leq |\Gamma'| \cdot 2^k \cdot 2^{k \log{k}}$. Thus, $|\Gamma| \in \order{2^{\order{k \log{k}}}}$. Moreover, $\Gamma$ took $\ordernoinput{2^{\order{k \log{k}}}}$-time to construct. Note that $M \subseteq S$ for all $\pdci \in \Gamma$ by construction. To complete our result, we must prove property (ii) for $\Gamma$ -- given a \domcol{} $\rchi'$ of $G$, there exists an extension $\rchi$ of a $\pdci \in \Gamma$ with $|\rchi| \leq |\rchi'|$.

Consider a \domcol $\rchi'$ of $G$. Let $\P = \{P_1, P_2 \dots P_\kappa\}$ be the partition of $M$ such that for two vertices $u$ and $v$, $\rchi'(u) = \rchi'(v)$ if, and only if, $u,v \in P_i$ for some $P_i \in \P$. We can, without loss in generality, assume that the vertices in $P_i$ are colored $c_i$. Let $\rchi^M \in \Gamma'$ be this partial coloring of $G$. Let $P'_1$ denote the set of colors used in both $M$ and $\Q$ by $\rchi'$. Let $P'_0$ be the rest of the colors. Let $\Q'$ be the set of cliques that do not dominate a color class in $P'_0$. By \Cref{Twin Cover Unique Colors Required}, each of these cliques must have a vertex which gets an unique color. Collect one such vertex from each $Q \in \Q'$. By \Cref{switching twins domcol}, we can assume that the set of vertices that our procedure colors with unique colors is exactly this set. Let $C_{\Q'}$ denote the set of colors used here. Let $M_0 \subseteq M$ denote the set of vertices which dominate a color class in $P'_0 \cup C_\Q'$. 

Since $\rchi'$ is a \domcol of $G$, every vertex in $G$, and in particular $M \setminus M_0$, dominates a color class. For a $v \in M \setminus M_0$, let $\delta'(v)$ be an arbitrary color class that $v$ dominates. Up to renaming, $\im{\delta'} \subseteq \widetilde{C}$. Consider the $\pdci \in \Gamma$ that the algorithm obtains by considering $\rchi^M \in \Gamma'$, $\P' = \{P'_0,P'_1\}$, and $\delta' \in \Delta'$. Consider a clique $Q \in \Q$ and let $c$ be a color that $\rchi'$ uses in $Q$. By construction, $c$ can be (or already is) used to color a vertex in $Q$ in an extension of $\pdci$. Since $\im{\pdci} \subseteq \im{\rchi'}$, this implies that there exists an extension $\rchi$ of $\pdci$ with $|\rchi| \leq |\rchi'|$. 

Therefore, by \Cref{dctcfpt}, there exists an algorithm to solve \DC in $\ordernoinput{2^{\order{k \log{k}}}}$-time.

% !TEX root = dcoloring.tex

\subsection{\CDC Parameterized by Twin Cover Size}\label{CD-TC}

We say that a clique $Q \in \Q$ is \textit{isolated} if $\closedneighbour{G}{Q_i} \cap M = \emptyset$. We have the following simple observations regarding isolated cliques.

\begin{observation}\label{isocliques}
	Consider an isolated clique $Q \in \Q$. Then, in any \classdcol of $G$, vertices in $Q$ must get unique colors.
\end{observation}
\begin{mycorollary}\label{reduction rule iso cliques}
	Let $\Q' = \{Q \in \Q \mid Q \text{ is isolated}\}$ and $\hat{\ell} = \sum_{Q \in \Q'}|Q|$. Then, $(G, \ell) \equiv (G - \Q', \ell - \hat{\ell})$.
\end{mycorollary}
We therefore have a reduction rule -- given a \CDC instance, we can consider the equivalent instance which has no isolated cliques. \textit{We assume that graphs in the rest of this subsection do not have any isolated cliques}. We make a few small observations before defining a \pclassdcol. Note that \Cref{switching twins} is similar to \Cref{switching twins domcol}.

\begin{observation}\label{dominated by M vertices}
	Let $\rchi$ be an arbitrary \classdcol of $G$ and let $c \in \im{\rchi}$ denote a color used by $\rchi$ outside $M$. Then, $c$ is dominated by a vertex in $M$.
\end{observation}
\begin{proof}
	Clearly, if $c$ is used in two cliques in $\Q$, it can only be dominated by a vertex in $M$. If $c$ is used in exactly one clique, then there must exist a vertex in $M$ which dominates it -- $G$ has no isolated cliques! Now, assume that $c$ is used both in $M$ and in exactly one clique $Q \in \Q$. Let $\rchi(u) = \rchi(v) = c$ where $u \in Q$ and $v \in M$. Assume that a $w \in V(\Q)$ dominates $c$. Then $w \in Q$ and $w \in \openneighbour{G}{v}$. However, this implies that $u \in \openneighbour{G}{v}$ (by \Cref{Twin Cover Clique Neighbourhood}), a contradiction. Thus, only a vertex in $M$ can dominate $c$. 
\end{proof}
\begin{observation}\label{switching twins}
	Let $\rchi$ be an arbitrary \classdcol of $G$ and let $u,v \in V(G)$ such that $\closedneighbour{G}{u} = \closedneighbour{G}{v}$. Then, $\rchi'$, the mapping which switches the colors assigned to $u$ and $v$ by $\rchi$ and retains the coloring of the rest of the vertices, is a \classdcol of $G$ with $|\rchi| = |\rchi'|$.
\end{observation}
\begin{mydefinition}
	Let $S \subseteq V(G)$. A \textit{\pclassdcol} is a mapping $\pddci \colon S \to C$ with the following properties: (i) it is a partial coloring of $G$, and (ii) every $c \in \im{\pddci}$ is dominated by some vertex in $V(G)$.
\end{mydefinition}
\begin{mydefinition}
	 A \classdcol $\rchi$ is an \textit{extension} of a \pclassdcol $\pddci$ if it agrees with $\pddci$ on $S$. An extension $\rchi$ of $\pddci$ is \textit{disjoint} if no vertex outside $S$ is colored using a $c \in \im{\pddci}$.
\end{mydefinition}
We will prove the following theorem, which establishes that an optimal disjoint extension of a \pclassdcol can be found quickly.
\begin{theorem}\label{disjoint extension in FPT time}
	Let $\pddci$ be a \pclassdcol of $G$ where $M \subseteq S$. Then, given an $\ell \in \N$, we can decide in $\ordernoinput{2^{\order{k\log{k}}}}$-time if there exists a disjoint extension $\rchi$ of $\pddci$ with $|\rchi| \leq \ell$.
\end{theorem}
For each $Q_i \in \Q$, let $b_i$ denote the number of vertices in $Q_i$ that are uncolored by $\pddci$. Let $\mathbf{b} = (b_1, b_2 \dots b_q)^T$ and let $\{v_1, v_2 \dots v_k\}$ be an arbitrary ordering of $M$. Consider the matrix $A \in M(q,k)$ with $A(i,j) = 1$ if $v_j \in \closedneighbour{G}{Q_i}$ and is 0 otherwise (similar to an adjacency matrix). Let $\mathbf{1}_k$ denote the vector in $\mathbb{R}^k$ containing only 1s. Consider the following Integer Linear Program (ILP): 
\[
\text{Minimize }\{\mathbf{1}_k \cdot \mathbf{x} \mid A\mathbf{x} \geq \mathbf{b},\ \mathbf{x} \in \Z^k,\ \mathbf{x} \geq 0 \}
\]
We refer the interested reader to Chapter 6 of \cite{Cygan2015} for a discussion on the use of ILPs in the world of parameterized algorithms. We describe the intuition behind designing this ILP. The $i$\textsuperscript{th} component of $\mathbf{x}$, denoted by $x_i$, represents the number of color classes outside $\im{\pddci}$ that $v_i \in M$ is required to dominate. Consider an arbitrary disjoint extension $\rchi$ of $\pddci$. By \Cref{dominated by M vertices}, every $c \in \im{\rchi} \setminus \im{\pddci}$ is dominated by some vertex in $M$. Since no two vertices in a clique can get the same color, for each $Q_i \in \Q$ the total sum of colors that the vertices in $\closedneighbour{G}{Q_i} \cap M$ dominate must be at least $b_i$ (the number of uncolored vertices in $Q_i$). This is exactly the $i$\textsuperscript{th} constraint in the linear program. We now formalize this notion below.

Note that the above program is feasible since $n \cdot \mathbf{1}$ is a feasible solution. Let $\ell^*$ denote the optimal value of the above program and $\ell' = |\pddci|$. 
\begin{restatable}{mylemma}{DisjointExtILPEq}\label{ILP equivalence}
	There exists a disjoint extension $\rchi$ of $\pddci$ with $|\rchi| \leq \ell$ if, and only if, $\ell^*  \leq \ell - \ell'$. 
\end{restatable}
\begin{proof}
	Consider a disjoint extension $\rchi$ of $\pddci$ with $|\rchi| \leq \ell$. By \Cref{dominated by M vertices}, every $c \in \im{\rchi} \setminus \im{\pddci}$ is dominated by some vertex in $M$. Make this choice arbitrarily and let $v_c$ denote this vertex. Let $x_j$ denote the number of colors that $v_j$ was chosen to dominate -- i.e., $x_j = |\{c \mid v_c = v_j\}|$. Since vertices in a clique must get different colors, for each clique $Q_i \in \Q$, $\sum_{j \mid v_j \in \closedneighbour{G}{Q_i}\cap M}x_j \geq b_i$. Therefore, $A\mathbf{x} \geq \mathbf{b}$. Moreover, as $\mathbf{x}$ satisfies the non-negativity and the integral constraints, $\mathbf{x}$ is a feasible solution of ILP. Now,
	\[
	\mathbf{1}_k \cdot \mathbf{x} = \sum_{i=1}^{k}x_i = |\im{\rchi} \setminus \im{\pddci}| = |\im{\rchi}| - |\im{\pddci}| \leq \ell - \ell' 
	\]  
	Thus, $\ell^* \leq 	\mathbf{1}_k \cdot \mathbf{x}  \leq \ell - \ell'$.
	
	Now, assume that $\ell^* \leq \ell - \ell'$. Let $\mathbf{x}$ be a feasible solution of the ILP such that $\mathbf{1}_k \cdot \mathbf{x}  \leq \ell - \ell'$. For $1 \leq j \leq k$, let $C_j = \{c^1_j, c^2_j \dots c^{x_j}_j\}$ be a set of colors. Color the uncolored vertices of $Q_i \in \Q$ by an arbitrary set of $b_i$-many colors from $\cup_{j \mid v_j \in \closedneighbour{G}{Q_i} \cap M} C_j$. Note that this is indeed possible -- $\sum_{j \mid v_j \in \closedneighbour{G}{Q_i}\cap M}x_j \geq b_i$ by the $i$\textsuperscript{th} constraint of the ILP. Let this coloring be $\rchi$. $\rchi$ is clearly a proper coloring of $G$. Moreover, a color in $C_j$ is dominated by $v_j$ in $\rchi$. Thus, $\rchi$ is a disjoint extension of $\pddci$. Finally, 
	\[
	|\rchi| \leq |\pddci| + \sum^k_{j=1}|C_j| \leq \ell' + (\ell - \ell') \leq \ell
	\]
	This completes the proof of this lemma.
\end{proof}
Since our ILP has $k$ variables, it can be solved in $\ordernoinput{2^{\order{k \log{k}}}}$-time \cite{Cygan2015}. Thus, \Cref{disjoint extension in FPT time} follows as a corollary of \Cref{ILP equivalence}.

Consider $\Gamma$, a collection of \pclassdcol{}s, with the following properties: 
\begin{enumerate}[(i)]
	\item For all $\pddci \in \Gamma$, $M \subseteq S$.
	\item Given a \classdcol{} $\rchi'$ of $G$, there exists an disjoint extension $\rchi$ of a $\pddci \in \Gamma$ with $|\rchi| \leq |\rchi'|$.
	\item $|\Gamma| = \order{2^{\order{k \log{k}}}}$ and $\Gamma$ can be constructed in $\ordernoinput{2^{\order{k \log{k}}}}$-time.
\end{enumerate}  
\begin{theorem}\label{tcfpt}
	Assume that a $\Gamma$ satisfying (i), (ii), and (iii) exists. Then, there exists an algorithm to solve \CDC in $\ordernoinput{2^{\order{k\log{k}}}}$-time, where $k$ is the size of a twin cover of the input graph. 
\end{theorem}
\begin{proof}
	Consider a collection of \pclassdcol{}s $\Gamma$ with the three properties elucidated above. By \Cref{disjoint extension in FPT time}, for each $\pddci \in \Gamma$ we can decide in $\ordernoinput{2^{\order{k\log{k}}}}$-time if there exists a disjoint extension $\rchi$ of $\pddci$ with $|\rchi| \leq \ell$. By the third property of $\Gamma$, we can decide in $\ordernoinput{2^{\order{k\log{k}}}}$-time if there exists a disjoint extension $\rchi$ of one of the \pclassdcol{}s in $\Gamma$. By its second property, $(G,\ell)$ is a \yes instance if, and only if, such a disjoint extension exists. 
\end{proof}
All that is left to do is to prove that such a collection exists. We do so below.
\subsubsection*{Constructing a $\Gamma$}
We construct this collection of \pclassdcol{}s in two steps. Let $\mathscr{P}$ denote the set of partitions of $M$. For each partition $\P = \{P_1, P_2 \dots P_\kappa\} \in \mathscr{P}$ ($\kappa$ is at most $k$), check if each part $P_i \in \P$ is independent and if there exists a $v \in V(G)$ which dominates $P_i$. If so, color all vertices in each $P_i$ by $c_i$. If not, reject this partition of $M$. Note that we have now obtained a collection $\Gamma'$ of $\order{2^{k \log{k}}}$-many \pclassdcol{}s of $G$. Moreover, constructing $\Gamma'$ required $\ordernoinput{2^{k \log{k}}}$-time. Note that $S = M$ for all $\pddci \in \Gamma'$ by construction. 

Consider a $\rchi^M \in \Gamma'$. Let $\mathscr{P}'$ be the set of ``partitions" of $\im{\rchi^M}$ into two sets. Partitions is in quotations here since we allow for $\im{\rchi^M}$ to be split into $\{\im{\rchi^M},\emptyset\}$. Consider a $\P' = \{P'_0,P'_1\} \in \mathscr{P}'$. We interpret colors in $P'_0$ as those that are only used in $M$ and those in $P'_1$ as those that can also be used in $G - M$. So far, our construction has been similar to \Cref{DCTCsect}. Define 
$
\Delta = \{\delta \colon P'_1 \to M \mid \delta(c) \text{ dominates } c \text{ for all }c \in P'_1\}
$.

Note that $|\Delta| = \order{2^{k \log{k}}}$. Consider a $\delta \in \Delta$. For each $c \in P'_1$, let $\Q_c$ denote the collection of cliques which are adjacent to $\delta(c)$ but not adjacent to vertices colored $c$ (refer \Cref{cdcoltwin1}). That is, 
$
\Q_c = \{Q \in \Q \mid \delta(c) \in \closedneighbour{G}{Q}, \closedneighbour{G}{Q} \cap {\rchi^M}^{-1}(c) = \emptyset\}
$. 

\begin{figure}[ht]
	\centering
	\includegraphics[width=0.67\textwidth]{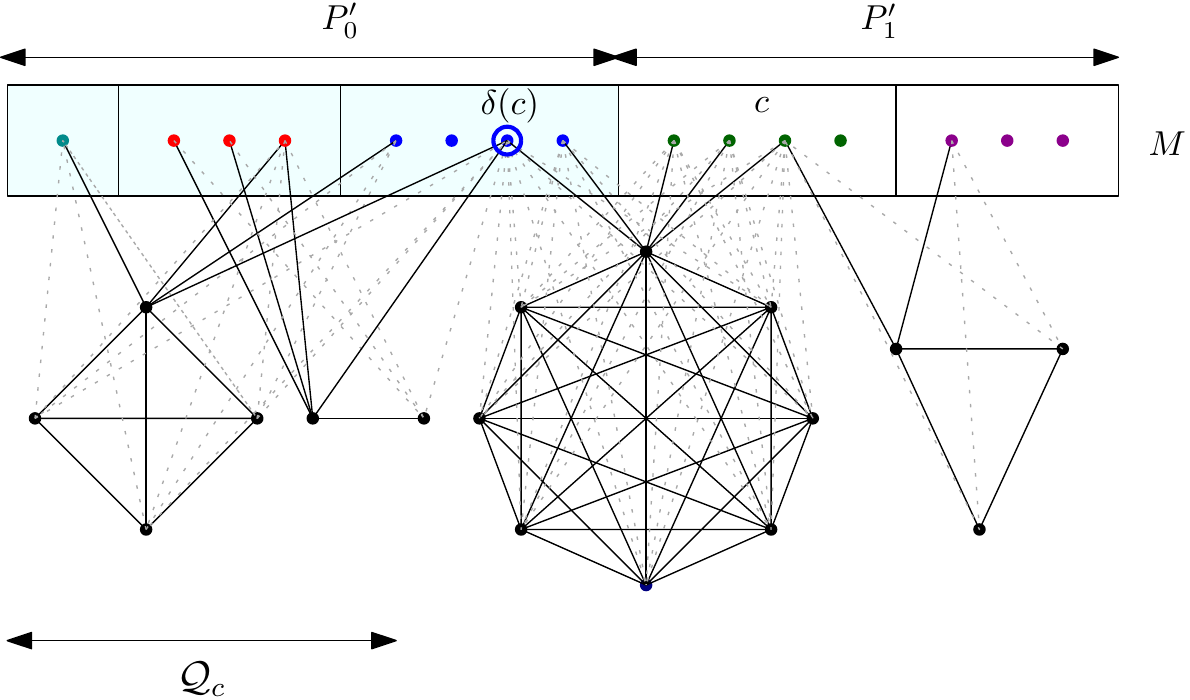}
	\caption{Given a $\rchi^M \in \Gamma'$, a $\{P'_0,P'_1\} \in \P'$, $\delta \in \Delta$, and a $c \in P'_1$ (the green vertices), we define $\Q_c$ to be the collection of cliques which are adjacent to $\delta(c)$ but not a vertex colored $c$.}
	\label{cdcoltwin1}
\end{figure}  

Now, for each $c \in P'_1$, color an arbitrary uncolored vertex (if such a vertex exists) of each $Q \in \Q_c$. Note that this indeed produces a \pclassdcol of $G$. Let $\Gamma$ be the collection of \pclassdcol{}s obtained by varying over $\pddci \in \Gamma'$, $\P' \in \mathscr{P}'$, and $\delta \in \Delta$.

\subsubsection*{$\Gamma$ Satisfies the Three Properties}

Since $|\mathscr{P}'| \leq 2^k$ and $|\Delta| \leq 2^{k \log{k}}$, $|\Gamma| \leq |\Gamma'| \cdot 2^k \cdot 2^{k \log{k}}$. Thus, $|\Gamma| \in \order{2^{\order{k \log{k}}}}$. Moreover, $\Gamma$ took $\ordernoinput{2^{\order{k \log{k}}}}$-time to construct. Note that $M \subseteq S$ for all $\pddci \in \Gamma$ by construction. To complete our result, we must prove property (ii) for $\Gamma$ -- given a \classdcol{} $\rchi'$ of $G$, there exists an disjoint extension $\rchi$ of a $\pddci \in \Gamma$ with $|\rchi| \leq |\rchi'|$. 

Consider a \classdcol $\rchi'$ of $G$. Let $\P = \{P_1, P_2 \dots P_\kappa\}$ be the partition of $M$ such that for two vertices $u$ and $v$, $\rchi'(u) = \rchi'(v)$ if, and only if, $u,v \in P_i$ for some $P_i \in \P$. We can, without loss in generality, assume that the vertices in $P_i$ are colored $c_i$. Let $\rchi^M \in \Gamma'$ be this \pclassdcol of $G$. Let $P'_1$ denote the set of colors used in both $M$ and $\Q$ by $\rchi'$. Let $P'_0$ be the rest of the colors. For each $c \in P'_1$, let $\delta(c) \in M$ denote an arbitrary vertex in $M$ that dominates $c$. Note that such a vertex exists due to \Cref{dominated by M vertices}. Consider the $\pddci \in \Gamma$ constructed due to the selection of the triplet $(\rchi^M, (P'_0,P'_1), \delta)$.

Assume there exists a clique $Q \in \Q$ such that $\rchi'(v) \in P'_1$ for all $v \in Q$. Then, by construction of $\pddci$, every vertex of $Q$ will be colored. Now, consider a clique $Q \in \Q$ such that $\rchi'(v) \in P'_1$ for all $v \in Q'$ where $\emptyset \subset Q' \subseteq Q$. Then, $\pddci$ colors at least $|Q'|$-many vertices of $Q$. By repeated application of \Cref{switching twins}, we can assume that $\pddci$ colors a superset of $Q'$. Thus, we can assume that $\rchi'$ does not use colors of $P'_1$ outside $S$. Let $\rchi$ be a mapping which agrees with $\pddci$ on $S$ and with $\rchi'$ on $V(G) \setminus S$. This is clearly a proper coloring of $G$. Moreover, a $c \in \im{\rchi} \subseteq \im{\rchi'}$ is dominated by the vertex that dominates $c$ in $\rchi'$. Thus, $\rchi$ is a disjoint extension of $\pddci$ where $|\rchi| \leq |\rchi'|$. Since $\rchi'$ was arbitrary, we have proved that $\Gamma$ satisfies (ii), the last remaining property that we were looking to prove. Thus, by \Cref{tcfpt}, we have an algorithm to solve \CDC in $\ordernoinput{2^{\order{k\log{k}}}}$-time.
\section{Parameterized by CVD Set Size}\label{cvd}
% !TEX root = dcoloring.tex
Consider a graph $G$. Recall that a subset $M$ of $V(G)$ is a cluster vertex deletion (CVD) set if $G - M$ is a cluster graph. Let $\Q = \{Q_i\}^q_{i=1}$ denote the components of the cluster graph $G - M$ in such an instance.  Note that each $Q_i \in \Q$ is a clique.  Consider an arbitrary ordering $\{v_1, v_2 \dots v_k\}$ of the vertices in $M$. For $0\leq j < 2^k$, we let $M_j$ denote the subset of $M$ which contains an $v_i \in M$ if, and only if, the $i$\textsuperscript{th} bit in the binary representation of $j$ is 1. Define, for each $i$ and $j$, $Q^j_i = \{v \in Q_i \mid \closedneighbour{G}{v} \cap M = M_j\}$. When we were discussing twin covers in \Cref{tc}, note that $Q_i = Q^j_i$ for some $j$ (\Cref{Twin Cover Clique Neighbourhood}). No such restriction exists for CVD sets. We will exploit this partitioning of $\Q$ to show that there exists an $\ordernoinput{2^{\order{2^k}}}$ and $\ordernoinput{2^{\order{2^kkq\log{q}}}}$-time algorithm which solves \DC and \CDC respectively.

% !TEX root = dcoloring.tex

\newcommand{\eN}{\mathcal{N}}
\newcommand{\V}{\mathcal{V}}
\newcommand{\E}{\mathcal{E}}

\makeatletter
\newcommand{\DESCRIPTION@original@item}{}
\let\DESCRIPTION@original@item\item
\newcommand*{\DESCRIPTION@envir}{DESCRIPTION}
\newlength{\DESCRIPTION@totalleftmargin}
\newlength{\DESCRIPTION@linewidth}
\newcommand{\DESCRIPTION@makelabel}[1]{\llap{#1}}%
\newcommand{\DESCRIPTION@item}[1][]{%
	\setlength{\@totalleftmargin}%
	{\DESCRIPTION@totalleftmargin+\widthof{\textbf{#1 }}-\leftmargin}%
	\setlength{\linewidth}
	{\DESCRIPTION@linewidth-\widthof{\textbf{#1 }}+\leftmargin}%
	\par\parshape \@ne \@totalleftmargin \linewidth
	\DESCRIPTION@original@item[\textbf{#1}]%
}
\newenvironment{DESCRIPTION}
{\list{}{\setlength{\labelwidth}{0cm}%
		\let\makelabel\DESCRIPTION@makelabel}%
	\setlength{\DESCRIPTION@totalleftmargin}{\@totalleftmargin}%
	\setlength{\DESCRIPTION@linewidth}{\linewidth}%
	\renewcommand{\item}{\ifx\@currenvir\DESCRIPTION@envir
		\expandafter\DESCRIPTION@item
		\else
		\expandafter\DESCRIPTION@original@item
		\fi}}
{\endlist}
\makeatother

%\newenvironment{altDescription}[1][\quad]
%{\begin{list}{}{
%			\renewcommand\makelabel[1]{\hfil\textsf{##1}}
%			\settowidth\labelwidth{\makelabel{#1}}
%			\setlength\leftmargin{\labelwidth+\labelsep}}}
%	{\end{list}}

%\newlist{altDescription}{description}{1}  % clone an existing list type
%\setlist[altDescription]{leftmargin=.2cm,labelsep=.2cm} 

\subsection{\DC Parameterized by CVD set size} \label{DC-CVD}

We associate each $Q_i \in \Q$ to a neighborhood vector $\eN_i$ of size $2^k$ as follows: $\eN_i(j)=|Q^j_i|$. That is, the entry $\eN_{i}(j)$ denotes the number of vertices of $Q_i$ with neighborhood $M_j$ in $M$. Let $\eN=\{\eN_i \mid 1\leq i\leq q\}$. Define an equivalence relation $\sim$ on $\eN$ as follows: $\eN_i \sim \eN_j$ if $\eN_i(p) \neq 0 \iff \eN_j(p) \neq 0$. Note that $\sim$ partitions $\eN$ into $\lambda \leq 2^{2^k}$ many equivalent classes. Let this partition be $\E$. We say a clique $Q_i$ is in $\E_p \in \E$ if $\eN_i \in \E_p$. For each $\E_p \in \E$, we define $\mathcal{NC}_i = \{M_j \subseteq M \mid \eN_i (j) \neq 0 \text{ for a }\eN_i \in \E_p\}$. Elements of $\mathcal{NC}_i$ are referred to as the \textit{neighborhood classes} of $\E_p$.

For an \domcol $\rchi$, denote the colors used in coloring the vertices in $M$ by $C^M(\rchi)$ and let $\rchi_M$ be the partial coloring of $M$. Without loss in generality, we can assume that $C^M(\rchi) \subseteq \{c_1,c_2 \dots c_k\}$. Let $C^0(\rchi) \subseteq C^M(\rchi)$ be the set of the colors that are only used in $M$. Consider any arbitrary set of colors $C^U(\rchi)$ such that each color in $c\in C^U(\rchi)$ is used to color only one vertex of a clique -- i.e. $\rchi^{-1}(c)=\{v\}$ for some $v\in Q_j$ and no two colors of $C^U(\rchi)$ is used to color the vertices of the same clique. With slight abuse of notation, we will drop $(\rchi)$ from the terms defined with respect to a coloring $\rchi$. The corresponding coloring will be clear from the context. We make a small observation.

\begin{observation}\label{Clique vertex domination}
	Let $v \in Q_i$ for some $Q_i \in \Q$. Then, $v$ cannot dominate a color class which is used to color a vertex $u \in Q_j$ for some $i \neq j$ in any \domcol of $G$.
\end{observation}

From \Cref{Clique vertex domination}, we know that the vertices in a clique $Q_i \in \Q$ either dominate a color in $C^U \cup C^0$ or a color which is used to color vertices only in $Q_i \cup M$ and not anywhere else. 
Consider a clique $Q_i$ such that no vertex of $Q_i$ is colored with a color in $C^U$. Then, every vertex dominates a color class in $C^M$. For such a clique $Q_i$, let $C_{Q_i}\subseteq C^M\setminus C^0$ denote a minimal set of colors such that every vertex in $Q_i$ dominates a color in $C_{Q_i} \cup C^0$. 
%For cliques not satisfying this premise, we let $C_{Q_i} = \emptyset$.
Let $C^1=\cup^q_{i=1} C_{Q_i}$ and define $C^2=C^M\setminus (C^0 \cup C^1)$.
Let $\alpha_\rchi$ be the number of cliques in which at least one color of $C^1$ was used.  Let the partition of $C^1$ among $\alpha_\rchi$ many cliques be $\P^1=\{\C^1_{1}, \C^1_{2} \cdots \C^1_{\alpha_\rchi}\}$. Consider the following observation which follows from the definition of the equivalence classes in $\E$.

\begin{observation} \label{Eq Classes C^0 domination}
	
	Let $\E_p \in \E$ be any equivalent class with  neighborhood classes $\mathcal{NC}_i$ and $\rchi$ be any \domcol. Then one of the two cases occur:
	\begin{description}
		\item[Case {\it (i)}] For all cliques $Q_i\in \E_p$, every vertex dominates a color class in $C^0$ 
		\item[Case {\it (ii)}] There exists a set of neighborhoods $\mathcal{NC}_i'\subseteq \mathcal{NC}_i$ such that for all $M'\in \mathcal{NC}_i'$ in each clique there exist a vertex with neighborhood $M'$ in $M$ which does not dominate a color in $C^0$. 
	\end{description}
	% 	For any equivalent class $\E_p \in \E$ and for any proper coloring $\rchi$, either  
	% 	
	% 	for all cliques there is a set of vertices with 
	% 	
	%	Let $\E_p \in \E$ and let $Q_i$ and $Q_j$ be two cliques in $\E_p$. Then, there exists a $v \in V(Q_i)$ which does not dominate a color class in $C^0$ if, and only if, there exists a $u \in V(Q_j)$ which does not dominate a color class in $C^0$.
\end{observation}

By \Cref{Eq Classes C^0 domination}, if in an $\E_i \in \E$, there exists a clique in this equivalence class which contains a vertex which does not dominate any color in $C^0$, every clique in $\E_i$ must be assigned a $\C^1_j \in \P^1$ to color its vertices or a $c \in C^U$ to color one of its vertices. For any equivalent class $\E_i$ such that every clique in $\E_i$ have at least one vertex which does not dominate any color in $C^0$, suppose we are given the following -- set of partitions $\P^1_i\subseteq \P^1(\rchi)$ assigned to $\E_i$; and for all $c \in \C^1_j$ in each partition $\C^1_j \in \P^1_i$, the neighborhood $\omega(c)$ in $M$ of the vertex that will get the color $c$. Denote such a set of neighborhoods by $N_M(c)$.

Observe that for each clique in $\E_i$ which is not assigned a partition of colors from $\P^1_i$, exactly one vertex must be colored with a color from $C^U$. We assume that we also know the neighborhood in $M$, $\M^U=\{M^U_{1},\cdots M^U_{\lambda}\}$ of the vertices in which these colors appear. Let $C=\{c_1\cdots c_{\ell}\}$ and $C^Q=C\setminus\{C^M\cup C^U\}$. Let $M_a\subseteq M$ be the set of vertices which dominate a color class in $C^0 \cup C^1$ and $M_b$ be the set which dominates a color class in $C^U$. Let $M_c = M \setminus (M_a \cup M_b)$. Without loss in generality, we can assume that vertices in $M_c$ dominate a color class in $\{c_1,c_2 \dots c_{2k}\}$. Let $M_d \subseteq M_c$ denote the vertices that dominate a $c \notin C^M$ and let $C_D\subseteq C^Q$ be the set of colors they dominate. Every color in $C_D$ must appear at least once in $G-M$. For each $c\in C_D$, assume we know the neighborhood in $M$, say $M^D(c)$, of the vertex in which it appears. Thus we can construct a $\delta \colon M\rightarrow \{c_1, c_2 \dots c_{2k}\}$ where $\delta(v)$ denotes a color that the vertex $v\in M$ dominates.

Given all this information, for each clique $Q\in \E_i$ and for any $\C^1_j$, we can determine whether it is possible to color the rest of the vertices with $C^2\cup C^Q$ where, for all $c\in \C^1_j$, one vertex with neighborhood $\omega(c)$ gets the color $c$. Similarly, we can compute, for any $Q\in \E_i$ and $M^U_{a}$, whether it is possible to color the rest of the vertices with $C^2\cup C^Q$ where one vertex with neighborhood $M^U_{a}$ gets a color from $C^U$. 

\begin{mylemma}\label{lemma:cliquecoloring}
	For any clique $Q_i$; given the set of colors $\C^1_j\subseteq C^1$, and, for all $c\in \C^1_j$, the neighborhood $\omega(c)$ in $M$ of the vertex that will get the color $c$, we can determine in polynomial time whether it is possible to color one vertex with neighborhood $\omega(c)$ with color $c$ for all $c\in \C^1_j$ and the rest of the vertices with colors from $C^2\cup C^Q$ such that for any vertex $m_a\in M$, the color $\delta(m_a)$ appears in the neighborhood of $m_a$.
\end{mylemma}

\begin{proof}
	Assign every color $c\in \C^1_j$ to one arbitrary vertex with neighborhood $\omega(c)$. If every vertex in $Q_i$ does not dominate at least one color from $C^0\cup \C^1_j$, report \no. Otherwise, create a bipartite graph $G_i$ with vertex bipartition $Q_i$ and $\C^1_j\cup C^2\cup C^Q$. For any color $c\in C^Q$, put an edge between $c$ and every vertex in $Q_i$. For $c\in \C^1_j$, let $M_c\subseteq M$ be the set of vertices which are assigned the color $c$ and $M_d\subseteq M$ be the set of vertices $v$ with $\delta(v)=c$. Put an edge between $c$ and $v$ if $v\notin \cup_{m_a\in M_c}\openneighbour{G}{m_a}$ and $v\in \cap_{m_d\in M_d}\openneighbour{G}{m_d}$. Return \yes if there exists a matching saturating $Q_i$ in $G_i$.
\end{proof}

Observe that there are some colors $C_{\epsilon}\subseteq C_D$ which must appear in the coloring of the vertices of the cliques to satisfy the domination requirement of some vertex in $M$. Assume that for each such color $c$, we know the neighborhood of the vertex, $\mu(c)$ which will be assigned the color. For any subset of colors $C_{\epsilon}'\subseteq C_{\epsilon}$, we can modify the construction of $G_i$ in \Cref{lemma:cliquecoloring} to determine whether there exists a coloring which satisfies the its premise and, in addition, every color of  $C_{\epsilon}'$ is used to color one vertex from $Q_i$. The modification is as follows: for each vertex $v$ with neighborhood $\mu(c)$, delete all edges incident on $v$ except the edge $(v,c)$. 

Consider any equivalent class $\E_i$ such that every clique in $\E_i$ has at least one vertex which does not dominate any color in $C^0$. For such a class, assume that we are given the set of partitions $\P^1_i\subseteq \P^1$ assigned to $\E_i$ and, for each partition $\C^1_j\in \P^1_i$ for all $c\in \C^1_j(\rchi)$, the neighborhood $\omega(c)$ in $M$ of the vertex that will get the color $c$. Denote such a set of neighborhoods by $N_M(c)$. We also have a subset $C_{\epsilon}'\subset C_{\epsilon}$. We would like to decide whether the cliques in $\E_i$ can be colored with $C^2\cup C^Q$ when each clique is assigned either a partition from $\P^1_i$ or a color from $C^U$. Moreover, we require that each color from $C_{\epsilon}'$ should be used to color at least one vertex.  From \Cref{lemma:cliquecoloring}, we know that for an arbitrary clique, an arbitrary part in $\P^1_i$, and any subset $C_{\epsilon}''\subseteq C_{\epsilon}'$; we can evaluate whether there exists a feasible coloring of $G$. In \fpt time we can guess the partitioning of $C_{\epsilon}'$ over different cliques and for each part $C_x$ in such partitioning, whether it will appear in a part of $\P^1_i$, or if it will appear in a clique which is assigned a color from $C^U$.

\begin{mylemma}
	For any equivalent class $\E_i$, given a partition $\P^1_i\subseteq \P^1$, a subset $C_u\subseteq C^U$ of cardinality $|\E_i|-|\P^1_i|$, a set of subset of  colors $\nu=\{\nu_1\cdots \nu_a\}$ where $\nu_j\subseteq C_{\epsilon}$ for $1\leq j\leq a$, a function $\alpha:\nu\rightarrow \P^1_i\cup C_u$ and for color $c\in \nu$ a function $\beta(c)\subseteq M$; we can determine, in polynomial time, if there exist a \pdc of $G$ such that all the vertices in $\E_i$ are colored with colors in $\P^1_i\cup C_u\cup C^2\cup C^Q$ and
	\begin{itemize}
		\item each color in $c\in \C^1_j\in \P^1_i$ appears in the same clique in a vertex with neighbourhood $\omega(c)$,
		\item each color $c\in \nu_j$ should appear in the same clique where $\alpha(c)$ appears in a vertex with neighborhood $\beta(c)$.
	\end{itemize} 
\end{mylemma}
\begin{proof}
	As in \Cref{lemma:cliquecoloring}, we can create a bipartite graph with bipartition $\E_i$ and $\{\P^1_i\cup C_u\}\times \nu$. That is, we create a vertex for each clique on one side. In the other side we have a vertex for each $(\zeta,\gamma)$  where $\zeta\in \{\P^1_i\cup C_u\}$ and either $\gamma$ is empty or $c\in \C^1_j\in \P^1_i$ where $\alpha(c)=\zeta$. We create an edge between $Q_c$ and $(\zeta,\gamma)$ if $(\zeta,\gamma)$ is feasible with respect to $Q_c$ (Lemma~\ref{lemma:cliquecoloring}). Report \yes if there exist a perfect matching in this graph report and \no otherwise.
\end{proof}

Therefore by trying all possible values for the following parameters, the rest of the problems can be solved in polynomial time.

\begin{tcolorbox}[colback=gray!5!white,colframe=gray!75!black]
	\begin{list}{}{
			\setlength{\leftmargin}{1.5cm}%% Make sure to pick a fixed value here. 
			\setlength{\itemindent}{0cm}
			\setlength{\labelwidth}{\leftmargin-\labelsep}
		}
		\item[$\rchi_M:$] Coloring of the modulator
		\item[$C^0(\rchi):$] Colors that have only appeared in $M$ 
		\item[$C^1(\rchi):$] Minimal set of colors $c$ which have appeared in $M$ and in a clique $Q_i$ such that there exist at least one vertex $v\in Q_i$ such that $c$ is the only color class that $v$ dominates
		\item[$\P^1(\rchi):$]  Partition of $C^1(\rchi)$ over different equivalent classes
		\item[$M_c(\rchi):$]  Neighborhood of the vertex $v$ in $G-M$ with $\rchi(v) = c$ for $c\in C^1(\rchi)$
		\item[$\V_i(\rchi):$]  Neighborhood vector of the equivalent class in which $\C^1_{i}(\rchi)$ appeared
		\item[$\eN_U^i(\rchi):$]  Neighborhood classes of $\E_i$ which are adjacent to a vertex colored $c \in C^U(\rchi)$
		\item[$\delta(\rchi):$]  Domination function for vertices which do not dominate a color class in $C^0(\rchi)\cup C^1(\rchi)\cup C^U(\rchi)$
		\item[$\gamma_\rchi:$]  Tuple with entries neighborhood of $v \in Q_i$ in $M$ and the neighborhood vector of $Q_i$ where $v \in M$ that dominates a color not in $C^M(\rchi) \cup C^U(\rchi)$
	\end{list}
	
\end{tcolorbox}

Therefore, we have the following theorem.

\begin{theorem}
	\DCCVD is \fpt. There exists an algorithm that solves it in $\ordernoinput{2^{\order{2^k}}}$-time.
\end{theorem}

\subsection{\CDC Parameterized by CVD Set Size}
We now discuss our much simpler algorithm for \CDC. By \Cref{isocliques} and \Cref{reduction rule iso cliques}, we can assume that there exists a $v \in Q_i$, for arbitrary $Q_i \in \Q$, with $\closedneighbour{G}{v} \cap M \neq \emptyset$. Pick one arbitrary vertex from each non-empty $Q^j_i$ and call this collection $M'$. Note that $|M \cup M'| \in 2^{\order{2^kq}}$. We have the following observation with the same flavor as \Cref{dominated by M vertices}.
\begin{observation}\label{dominated by M U M' vertices}
	Let $\rchi$ be an arbitrary \classdcol of $G$ and let $c \in \im{\rchi}$ denote a color used by $\rchi$ outside $M$. Then, $c$ is dominated by a vertex in $M \cup M'$.
\end{observation}
Note that the proof of \Cref{disjoint extension in FPT time} and the construction of \pclassdcol{}s that followed hinged entirely on \Cref{dominated by M vertices} and not the fact that $M$ was a twin cover. Thus, using  \Cref{dominated by M U M' vertices}, along with the ideas in \Cref{CD-TC}, we have the following theorem.
\begin{theorem}
	There exists an algorithm to solve \CDCol in $\ordernoinput{2^{\order{2^kkq\log{q}}}}$-time, where $k = |M|$, where M is a CVD set of the input graph and $q$ is number of connected components in $G - M$. 
\end{theorem}
% !TEX root = dcoloring.tex
\section{Lower Bounds}\label{lowerbounds}
In this section, we establish some lower bounds for \DC and \CDC with the parameters discussed in the previous sections. Consider a graph $G$. Construct a graph $G'$ from $G$ as follows: $V(G') = V(G) \cup \{v_0\}$ for a new vertex $v_0$ and $E(G') = E(G) \cup \{(v_0,v) \mid v \in V(G)\}$. By \cite{Gera2006,Merouane2015} we have the following.
\begin{observation}
	$G$ has a proper coloring using $\ell$-many colors if, and only if, $G'$ has a \domcol or a \classdcol using $(\ell+1)$-many colors.
\end{observation}
Note that $V(G)$ is a clique modulator and a twin cover in $G'$. Therefore, if \DC and \CDC can be solved in $\ordernoinput{f(k)}$-time, where $k$ is the size of the modulator and $f$ is some computable function, then \Col can be solved in $\order{f(n)}$-time where $n$ is the number of vertices in an instance of this problem. It is known, under the Exponential-Time Hypothesis (\textsf{ETH}), that \Col cannot be solved in $2^{o(n)}$-time (refer \cite{Fomin2010,Cygan2015}). We state this discussion below as a lemma.
\begin{mylemma}\label{ETH lowerbound}
	\DC and \CDC do not admit $\ordernoinput{2^{o(k)}}$-time algorithms unless \textup{\textsf{ETH}} fails. Here, $k$ is the size of a clique modulator or a twin cover of the input graph.
\end{mylemma}
The proofs of the following results utilize a relationship between \HS and \DC that we now establish. 
\begin{restatable}{theorem}{SCC}\label{SCCDC}
	Unless the Set Cover Conjecture is false,  \DC does not admit a $\ordernoinput{{(2-\epsilon)}^{k}}$-time algorithm where $k$ is the size of a CVD set.
\end{restatable}
\begin{restatable}{theorem}{nopolykernel}\label{NPKDC}
	Unless $\npsubconp$ \DC, parameterized by the size of a CVD set, does not admit a polynomial kernel.
\end{restatable}
Let $U$ be a set of elements and $\F \subseteq \powerset{U}$. A set $S \subseteq U$ is called a \textit{hitting set} of $\F$ if, for all $F \in \F$, $F \cap S \neq \emptyset$. We define the well known \HS problem below.
\defproblem{\HS}
{A set of elements $U$ and $\F \subseteq \powerset{U}$; an integer $\kappa$}
{Does there exist a hitting set $S \subseteq U$ of $\F$ with $|S| \leq \kappa$?}
We let $\HSinstance$ denote an instance of \HS and assume that $|U| = n$. We will construct an instance $\DCinstance$ of \DC from an instance $\HSinstance$ of \HS which has the following property: $\HSinstance$ is a \yes instance if, and only if, $\DCinstance$ is a \yes instance where $\ell = n+2$.

For each $F \in \F$, we have a vertex $m_F \in V(G)$ and for each $u \in U$, we have a corresponding vertex $v_u \in V(G)$. Moreover, in $V(G)$, we have a set $Q_2$ of $(n-\kappa)$-many vertices and two special vertices $m_1$ and $m_2$. We now define the edges of $G$. All the vertices of $Q_1 = \{v_u \mid u \in U\}$ are connected to each other as are all vertices in $Q_2$. The vertices $m_1$ and $m_2$ are adjacent to all vertices in $Q_1 \cup Q_2$ and to each other. For an $F \in \F$, $\openneighbour{G}{m_F} = \{v_u \in Q_1 \mid u \in F\}$. Let $M_{\F} = \{m_F \mid F \in \F\}$. Refer to \Cref{HS to DC Figure} for a pictorial representation of $G$. 

Note that $M = M_{\F} \cup \{m_1,m_2\}$ is a cluster vertex deletion set of $G$ and $G - M$ consists of the two cliques $Q_1$ and $Q_2$. Let $\ell = n+2$ and $k = |M| = |\F| + 2$. First, note that $G$ cannot be colored with fewer than $(n+2)$-many colors since $Q_1 \cup \{w_1,w_2\}$ is a clique of size $n+2$. The main result of this section is the following.

\begin{theorem} \label{Equivalence between DCCVD and HS}
	$\HSinstance$ is a \yes instance of \HS if, and only if, $\DCinstance$ is a \yes instance of \DC.
\end{theorem}
It is known, due to \cite{Dom2014}, that \HS does not admit a polynomial kernel unless $\npsubconp$. Moreover this also implies, under the Set Cover Conjecture \cite{Cygan2016}, that \DC does not admit any $\ordernoinput{{(2-\epsilon)}^{k}}$-time algorithm where $k$ is the size of a CVD set. Therefore, \Cref{NPKDC,SCCDC} follow from \Cref{Equivalence between DCCVD and HS}. We prove this theorem in two steps: first we show that $\HSinstance$ is a \yes instance of \HS if, and only if, there exists a \domcol $\rchi$ of $G$ with $|\rchi| = n+2$ such that all vertices in $M \setminus \{m_2\}$ get the same color and then show that there exists a \domcol with this property of optimal size. We prove these two results as the final two lemmas of the paper below.
\begin{figure}[ht]
	\centering
	\includegraphics[width=0.48\linewidth]{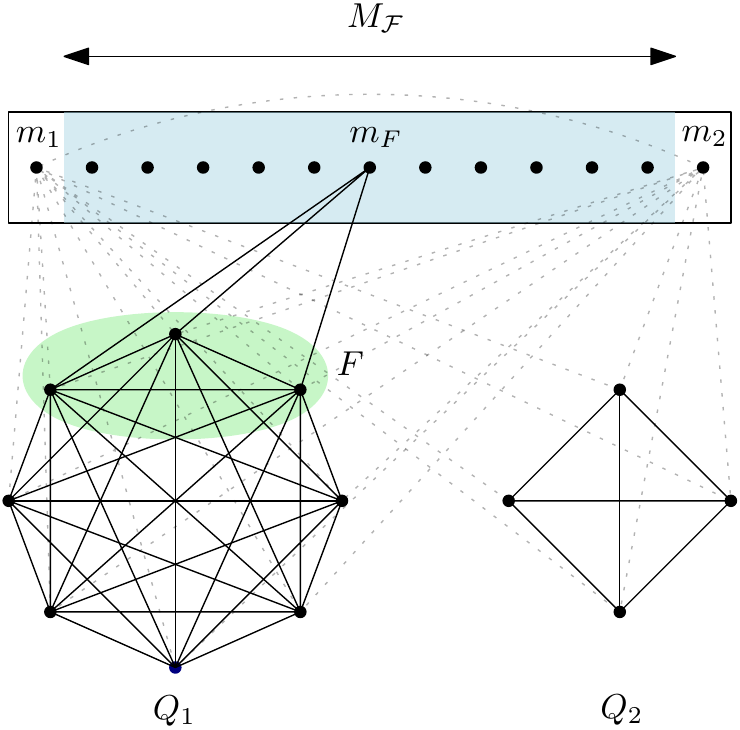}
	\caption{Illustration of construction of $\DCinstance$ from $\HSinstance$.}
	\label{HS to DC Figure}
\end{figure}
\begin{mylemma}
	$\HSinstance$ is a \yes instance of \HS if, and only if, there exists a \domcol $\rchi$ of $G$ with $|\rchi| = n+2$ such that $\rchi(m_1) = \rchi(m_F)$ for all $m_F \in M_{\F}$.
\end{mylemma}
\begin{proof}
	Assume that $\HSinstance$ is a \yes instance of \HS. Then, there exists a set $S \subseteq U$ with $|S| \leq \kappa$ which hits every set of $\F$. We define a coloring $\rchi$ of $G$ as follows. Color each vertex in $Q_1 \cup \{m_1,m_2\}$ a different color and let $\rchi(m_F) = \rchi(m_1)$ for all $m_F \in M_{\F}$. Let $Q^S_1 = \{v_u \in Q_1\mid u \in S\}$. Then, $|Q^S_1| = |S| \leq \kappa$ and hence $|Q_1 \setminus Q^S_1| \geq n - \kappa$. Color the $(n-\kappa)$-many vertices in $Q_2$ using the colors of in $\rchi(Q_1 \setminus Q^S_1)$. 
	
	Note that $\rchi$ is indeed a proper coloring of $G$ -- for if $\rchi(u) = \rchi(v)$ for some $u$ and $v$, then either $u,v \in M \setminus \{m_2\}$ or (without loss in generality) $u \in Q_1$ and $v \in Q_2$. In either case, $(u,v) \notin E(G)$. Moreover, vertices in $Q_1 \cup Q_2 \cup \{m_1,m_2\}$ dominate the color class $\rchi(m_2)$. Consider a vertex $m_F \in M_{\F}$. Since there exists a $s \in S \cap F$, $x_F$ is adjacent to $v_s \in Q_1$. Since $\rchi(v_s)$ is not used to color any other vertex in the graph, $m_F$ dominates this color class. Since $|\rchi| = n+2$ and $\rchi(m_1) = \rchi(m_F)$ for all $m_F \in M_{\F}$, we have the forward direction of this proof.
	
	Now, let $\rchi$ be a \domcol of $G$ with $|\rchi| = n+2$ where $\rchi(m_1) = \rchi(m_F)$ for all $m_F \in M_{\F}$. Since $Q_1 \cup \{m_1,m_2\}$ is a clique of $G$, $|\rchi(Q_1 \cup \{m_1,m_2\})| = n+2$. Moreover, the $(n-\kappa)$-many vertices in $Q_2$ cannot be colored with either $\rchi(m_1)$ or $\rchi(m_2)$ as $\closedneighbour{G}{m_1} = \closedneighbour{G}{m_2} \supseteq Q_2$. Thus, $\rchi$ must use $(n-\kappa)$-many colors of $\rchi(Q_1)$ to color the vertices of $Q_2$. 
	
	Moreover as every vertex in $M_{\F}$ is colored with $\rchi(m_1)$, they must dominate a color class that is only used in $Q_1$. By the final sentence of the previous paragraph, there are at most $\kappa$-many such vertices in $Q_1$. Let $Q^S_1$ denote this set of vertices and $S = \{u \in U \mid v_u \in Q^S_1\}$. Each vertex in $M_{\F}$ must be adjacent to at least one vertex in $Q^1_S$. By construction of $G$ this implies that $S \cap F \neq \emptyset$ for all $F \in \F$. Since $|S| \leq \kappa$, $\HSinstance$ is a \yes instance of \HS.
\end{proof}
\begin{mylemma}
	$\DCinstance$ is a \yes instance of \DC if, and only if, there exists a \domcol $\rchi$ of $G$ with $|\rchi| = n+2$ such that $\rchi(m_1) = \rchi(m_F)$ for all $m_F \in M_{\F}$.
\end{mylemma}
\begin{proof}
	Assume that $\DCinstance$ is a \yes instance of \DC. Then, there exists a \domcol $\overline{\rchi}$ of $G$ with $|\overline{\rchi}| = n+2$. By construction of $G$, a vertex $m_F \in M_{\F}$ must either dominate the color class $\overline{\rchi}(m_F)$ or $\rchi(v_u)$ for some $v_u \in Q_1$. Note that, in either case, exactly one vertex is colored with the color $m_F$ dominates. Let $\delta(m_F)$ denote that vertex. Let $M_a = \{m_F \in M_{\F} \mid \delta(m_F) = m_F\}$. We construct a coloring $\rchi$ of $G$ from $\overline{\rchi}$ as follows: color every vertex in $M_{\F}$ $\overline{\rchi}(m_1)$, for each $m_F \in M_a$ color an arbitrary vertex in $\openneighbour{G}{m_F}$ with $\overline{\rchi}(m_F)$ (if $\openneighbour{G}{m_F} = \emptyset$, then $|\overline{\rchi}| > n+2$), and color the rest of the vertices with the same color that was used by $\overline{\rchi}$.
	
	Note that the only color that is used multiple times by $\rchi$ but not by $\overline{\rchi}$ is $\rchi(m_1) = \overline{\rchi}(m_1)$. Since $\rchi^{-1}_d(m_1) = M_{\F} \cup \{m_1\}$ is an independent set, $\rchi$ is a proper coloring of $G$. Moreover, note that every vertex in $V(G) \setminus M_{\F}$ dominates $\rchi(m_2)$ as $m_2$ is the only vertex colored with $\rchi(m_2)$ by $\rchi$. As every vertex $m_F \in M_{\F}$ has a $v \in \openneighbour{G}{m_F}$ which gets a unique color in $\rchi$, they dominate a color class as well. Hence, $\rchi$ is a \domcol of $G$ with $|\rchi| = n+2$ such that $\rchi(m_1) = \rchi(m_F)$ for all $m_F \in M_{\F}$. Since the other side of the claim is trivial, this completes the proof of this lemma.
\end{proof}
This proves \Cref{Equivalence between DCCVD and HS} and therefore, \Cref{SCCDC,NPKDC}.

\section*{Conclusion and Future Work}
% !TEX root = dcoloring.tex
In this paper, we introduced the study of structural parameterizations of \DC and \CDC. There were three parameters of interest: clique modulator size, twin cover size, and CVD set size. We designed randomized algorithms, based on an inclusion-exclusion based polynomial sieving method, running in $\ordernoinput{c^k}$-time (for $c = 16$ for \DC and $c = 2$ for \CDC),  where $k$ was the clique modulator size. For twin cover size, we had $\ordernoinput{2^{\order{k\log{k}}}}$-time algorithms for both problems. These algorithms hinged on solving \PreCol-type problems quickly for \domcol{}s and \classdcol{}s. Finally, for the smallest parameter, CVD set size, we described an involved a double exponential-time branching algorithm for \DC. For \CDC, our (much simpler) algorithm also depended on the number of cliques in the cluster graph obtained by deleting a CVD set. We list a few interesting open questions below:

\begin{enumerate}[(i)]
	\item We have shown that under ETH, we cannot do better than a $\ordernoinput{2^k}$-time algorithm for our two problems when $k$ is the clique modulator size. While we have achieved this bound for \CDC, our algorithm for \DC is (possibly) suboptimal. Crucially, the current bottleneck in our algorithm's running time is the exact $\ordernopoly{4^n}$-time algorithm for \DC we designed in \Cref{exactalg}. If this can be bettered to $\ordernopoly{2^n}$ (which, we believe, is possible), our algorithm will run in $\ordernoinput{4^k}$-time.
	\item Does a faster \fpt algorithm for \DC parameterized by CVD set size exist? We have shown that \CDC, parameterized by CVD set size and the number of cliques which remain, is \fpt. Is there a branching algorithm, similar to the one in \Cref{DC-CVD} for \DC, which is \fpt with respect to \textit{only} CVD set size? Is \CDC, instead, \wih with respect to this parameter?
	\item While we have shown strong positive results in this work, most lower bounds we designed are (possibly) loose. The constructions we use in \Cref{lowerbounds} all use a few (at most two) cliques -- we believe that more involved constructions can give much tighter bounds, especially for parameterization with CVD set size.
\end{enumerate}
\section*{Acknowledgments}
% !TEX root = dcoloring.tex
The authors would like to thank Shravan Muthukumar for their preliminary discussions on \DC and \CDC. They are grateful to the anonymous reviewers of the 40th International Symposium on Theoretical Aspects of Computer Science (STACS 2023) for their in-depth comments on the first version of this work which only discussed \DC. 
\section*{On Rolf Niedermeier}
% !TEX root = dcoloring.tex

The last author fondly remembers his long association with Rolf Niedermeier as one of the early converts to parameterized complexity. His first meeting was in the first Dagstuhl on parameterized complexity and he was glad to know that Rolf was inspired by his early survey in the area. Since then both have had multiple meetings including a joint Indo-German project when Rolf was in Jena. He fondly remembers his interesting visits to Jena and Rolf's visit to Chennai. He is also particularly inspired by Rolf's constant quest for application areas where parameterized complexity can be applied.

\bibliographystyle{splncs04}
\bibliography{Bibliography.bib}

\end{document}